\documentclass[journal,
	]{IEEEtran}%
\pdfoutput=1

\usepackage[T1]{fontenc} 
\usepackage[english]{babel} 

\usepackage[usenames,dvipsnames]{xcolor}

\usepackage{cite} 

\usepackage{booktabs}

\usepackage{tikz}
\usetikzlibrary{backgrounds}
\definecolor{lightblue}{rgb}{0.3,0.7,0.95}

\usepackage{amssymb} 
\usepackage{mathtools} 
\usepackage{amsmath} 
\usepackage{amsfonts}
\usepackage{amsthm} 

\usepackage{algorithmic} 
\usepackage{algorithm} 
\usepackage{balance}
\usepackage{complexity}

\usepackage[colorlinks=true]{hyperref}

\usepackage[normalem]{ulem} 
\usepackage{cancel} 
\newtheorem{theorem}{Theorem}
\newtheorem{corollary}[theorem]{Corollary}
\newtheorem{lemma}[theorem]{Lemma}
\newtheorem{proposition}[theorem]{Proposition}
\newtheorem{observation}[theorem]{Observation}
\newtheorem*{observation*}{Observation}
\theoremstyle{definition}

\newtheorem{definition}{Definition}

\newcommand{\e}{\ensuremath\mathrm{e}}
\renewcommand{\i}{\ensuremath\mathrm{i}}

\DeclareMathOperator{\vecmap}{vec}

\DeclareMathOperator{\Id}{Id}
\DeclareMathOperator*{\argmin}{arg\,min}
\DeclareMathOperator{\supp}{supp}

\newcommand{\CC}{\mathbb{C}}
\newcommand{\RR}{\mathbb{R}}
\newcommand{\KK}{\mathbb{K}}

\newcommand{\ZZ}{\mathbb{Z}}
\newcommand{\NN}{\mathbb{N}}

\renewcommand{\vec}[1]{\mathbf{#1}}

\newcommand{\norm}[1]{\left\Vert #1 \right\Vert} 
\newcommand{\normb}[1]{\bigl\Vert #1 \bigr\Vert} 
\newcommand{\normB}[1]{\Bigl\Vert #1 \Bigr\Vert} 
\newcommand{\snorm}[1]{\norm{#1}} 

\newcommand{\snormb}[1]{\normb{#1}}

\newcommand{\braket}[2]{\left\langle #1, #2 \right\rangle}

\renewcommand{\A}{\vec A} 
\newcommand{\B}{\vec B} 
\newcommand{\x}{\vec x} 
\newcommand{\y}{\vec y} 
\newcommand{\z}{\vec z} 
\newcommand{\ev}{\vec e} 

\newcommand{\proj}{\ensuremath\kern -0.12em\rfloor\kern -0.06em} 

\newlang{\bhtp}{HiHTP} 
\newlang{\HiHTP}{HiHTP} 
\newlang{\htp}{HTP} 
\newlang{\HTP}{HTP} 
\newlang{\GOMP}{GOMP}
\newlang{\HiLasso}{HiLasso}

\newcommand{\tree}[1]{\ensuremath \vec #1}
\DeclareMathOperator{\children}{children}
\DeclareMathOperator{\parent}{parent}
\DeclareMathOperator{\leaves}{leaves}

\newcommand{\fu}{Dahlem Center for Complex Quantum Systems, Freie Universit\"{a}t Berlin, Germany}
\newcommand{\ud}{Institute for Theoretical Physics, Heinrich Heine University D\"{u}sseldorf, Germany}
\newcommand{\ug}{Institute of Theoretical Physics and Astrophysics, University of Gda\'{n}sk, Poland}
\newcommand{\fuMCS}{Department of Mathematics and Computer Science, Freie Universit\"{a}t Berlin, Germany}
\newcommand{\mv}{Institute for Mathematical Sciences, University of Gothenburg and Chalmers University of Technology, Sweden}

\title{Reliable recovery of hierarchically sparse signals for Gaussian and Kronecker product measurements}

\author{Ingo Roth, Martin Kliesch, Axel Flinth, Gerhard Wunder, and Jens Eisert
	\thanks{Parts of this work were presented at ``The first colloquium on the priority programme compressed sensing in information processing (CoSIP)'' in Aachen, Germany, the ``International Traveling Workshop On Interactions Between Sparse Models and Technology, iTWIST'' \cite{RothEtAl:iTwist:2016} in Aalborg, Denmark, in 2016 and the 56th Annual Allerton Conference on Communication, Control, and Computing \cite{RothEtAl:2018} in Monticello, Illinois, USA in 2018. 
	I.\ Roth and J.\ Eisert	are with the \fu; 
		M.\ Kliesch is with the \ud{} and the \ug; A.\ Flinth is with the \mv; 
		G.\ Wunder and J.\ Eisert are with the \fuMCS. (i.roth@fu-berlin.de)
		}
}

 \makeatletter 
  \hypersetup{pdftitle = {Reliable recovery of hierarchically sparse signals for Gaussian and Kronecker product measurements},
	     pdfauthor = {Ingo Roth, Martin Kliesch, Axel Flinth, Jens Eisert, Gerhard Wunder},
 	     pdfsubject = {Hierarchically sparse signals},
 	     pdfkeywords = {Compressive sensing, model-based, 
 				     support recovery, hierarchically sparse,
 				     block sparse, level sparse, hard thresholding pursuit,
 				     iterative hard thresholding, sparse recovery, mobile communications
 				     }
 	    }
 \makeatother

\begin{document}

\maketitle


\begin{abstract} 
	We propose and analyze a solution to the problem of recovering a block sparse signal with sparse blocks from linear measurements. Such problems naturally emerge inter alia
	in the context of 
	mobile communication, 
	in order to meet the scalability and low complexity requirements of 
	massive antenna systems and massive machine-type communication.
	We introduce a new variant of the Hard Thresholding Pursuit (\htp) algorithm 
	referred to as \HiHTP. We provide both a proof of convergence and 
	a recovery guarantee for noisy Gaussian measurements that 
	exhibit an improved asymptotic scaling in terms of the sampling complexity in comparison with the usual {\htp} algorithm. 
	Furthermore, hierarchically sparse signals and Kronecker product structured measurements naturally arise together in a variety of applications. We establish the efficient reconstruction of hierarchically sparse signals from Kronecker product measurements using the \HiHTP\ algorithm. 
	Additionally, we provide analytical results that connect our recovery conditions to generalized coherence measures.
	Again, our recovery results exhibit substantial improvement in the asymptotic sampling complexity scaling over the standard setting. Finally, we validate in numerical experiments that for hierarchically sparse signals, \HiHTP\ performs  significantly better compared to \htp . 
\end{abstract}


\section{Introduction}
An important task in the recovery of signals is to approximately reconstruct a vector $\vec x \in \CC^{d}$ from $m$ noisy linear measurements
 \begin{equation}\label{eq:genlinprob}
   \y = \A\x + \ev 
   \ \in \CC^m ,
 \end{equation}
 where $\A\in \CC^{m \times d}$ is the measurement map and $\ev\in \CC^m$ denotes additive noise. 
Such problems arise in many applications, e.g., in image processing, acoustics, radar, machine learning, quantum state tomography, and mobile communication.
Particularly interesting is the case where the number of measurements $m$ is much smaller than the dimension of the signal space and where $\x$ has some known 
structure. 
The machinery of compressed sensing (CS) provides powerful tools for an efficient, stable, and unique reconstruction of $\x$ from $\y$ and $\A$. 
For many instances, this idea works extraordinarily well and is most prominently investigated for the case where the vector $\x$ is $s$-sparse, i.e., $\x$ has at most $s$ non-zero entries,
 see Ref.\ \cite{FouRau13} and references therein.

\subsection{Problem formulation}
In this work, we investigate a natural setting with more structure than mere $s$-sparsity
in which even superior performance of reconstruction is to be expected.
\begin{definition}
A vector $\x\in \CC^{Nn}$ (and the corresponding supports) is \emph{$(s,\sigma)$-hierarchically sparse} or simply $(s,\sigma)$-sparse
if $\x$ consists of $N$ blocks each of size $n$, where at most $s$ blocks have non-vanishing entries and additionally each non-zero block itself is $\sigma$-sparse.
\end{definition}
This structure can be regarded as a simple instance of a larger class of hierarchical sparsity structures which have certain sparsities for nested levels of groupings  and has been introduced in Refs.\ \cite{SprechmannEtAl:2010,FriedmanEtAl:2010, SprechmannEtAl:2011, SimonEtAl:2013}. We shall provide a precise generalization of our definition in Section~\ref{sec:generalHierachicalSparsity}. 
Hierarchically sparse signals arise in a variety of different applications, see Refs.\ \cite{DigheAsaeiBourlard:2016,DaoEtAl:2016, LiuEtAl:2016} for recent examples. 
It can further be seen as a combination of block sparsity \cite{EldarMishali:2009a,EldarMishali:2009b} and level sparsity \cite{AdcockEtal:2013,LiAdcock:2016}.
This study is motivated by the desire to identify model-based sparsity structures \cite{BarCevDua10} that (i) allow for efficient, stable and reliable recovery algorithms with analytical performance guarantees and (ii) are re\-le\-vant in concrete technological applications.

\subsection{Contributions of this work}\label{sec:contribution}
We propose new low-complexity thresholding algorithms for the recovery of hierarchically sparse vectors from incomplete noisy linear measurements and provide a stable and robust recovery guarantee based on generalized RIP conditions, 
so-called HiRIP. 
For Gaussian measurements we provide RIP bounds for a large class of hierarchical sparsity patterns including the simple case of $(s,\sigma)$-sparsity. Therefore, the proposed algorithm has an improved sampling complexity compared to standard CS approaches. The algorithm is compared to standard CS algorithms and convex optimization techniques for hierarchically sparse structures. For Gaussian and subsampled Fourier measurements an improved sampling complexity is found numerically for the proposed algorithm. Furthermore, we extend our analysis to structured Kronecker product measurements which naturally arise in the context of hierarchically sparse signals identifying each matrix in the Kronecker product with a level in the hierarchy. Our analysis implies that sensing matrices resulting from Kronecker products exhibit HiRIP, provided each of the constituent matrices satisfy the (standard) RIP. Putting it firmly, the HiRIP is attainable, while RIP may actually not,
relaxing substantially the RIP requirements of each constituent matrix. We also relate the HiRIP-framework to generalized coherence conditions from the literature, both for general as well as for Kronecker measurements.

\subsection{Selected concrete applications} \label{subsec:application}
Several large-scale signal processing problems in communications do actually exhibit 
hierarchically sparse signals and \HiHTP's simplicity, scalability, and performance guarantees makes its application attractive in a variety of settings. 
Two very recent prominent examples are \emph{massive MIMO (mMIMO)} 
and \emph{massive machine-type communications (mMTC)}. In Refs.~\cite{Wunder2019_TWC, WunderEtAl:WSA:2018, Bazzi2019_ACCESS} we have  considered the
application of \HiHTP\ to mMIMO channel estimation with a uniform linear array antenna configuration.
Such estimation problems can be posed as estimating a sparse matrix $\mathbf{X}\in \CC^{n\times N}$ from subsampling the rows and columns of
\begin{equation}
\mathbf{H}=\mathcal{F}^{(n)}\mathbf{X}(\mathcal{F}^{(N)})^\ast,
\end{equation}
where $\mathcal{F}^{(N)}\in \CC^{N\times N}$ and $\mathcal{F}^{(n)}\in \CC^{n\times n}$ are the matrix representations of the discrete Fourier transform of size $N$ and $n$
representing frequency and angular domain of a wide-band massive MIMO channel, respectively.
We denote by $\bar{\vec A}$ and $\vec A^\ast$ the complex conjugate and Hermitian transpose of a matrix $\vec A$, respectively. 
The $s$ non-vanishing entries within the  matrix $\mathbf{X}$ can then be interpreted as the channel coefficients of resolvable paths of
which the corresponding row indices encode their discrete excitation times $\tau$ and the column indices 
 encode their discrete angles $\theta$. In a physical environment one can typically assume that if the resolution, i.e.\ $N$ and $n$,
is large enough there is only one path per angle, i.e. there is
only one non-zero coefficient per column.
Introducing subsampling matrices $\Phi^{(\tau)},\Phi^{(\theta)}$ in appropriate dimensions and expressing the equation in column-vectorized form gives
\begin{equation}
\mathrm{vec}(\mathbf{H})=(\Phi^{(\theta)}\bar{\mathcal{F}}^{(N)})\otimes (\Phi^{(\tau)}\mathcal{F}^{(n)}) \mathrm{vec}(\mathbf{X}).
\end{equation}
Here, $\mathrm{vec}(\mathbf{H})$ is
 $(s,1)$ hierarchically sparse which is precisely the structure that we can handle with the results
in this paper. Generalizing the described mMIMO setting to multiple users with sporadic activity gives rise to another sparse hierarchy level. We refer to Ref.~\cite{Wunder2019_TWC
} for details.

Another major application is the support of mMTC sporadic traffic \cite{Bockelmann2018_ACCESS,WunderEtAl:mMTC:asilomar:2017,WunderEtAl:MTC:2018} where a massive number of devices
access in an uncoordinated manner the same communication resource.
Hierarchically sparse structures naturally emerge when sparse device activity and the sparse channel impulse responses (CIRs) are jointly exploited.
Recent protocols \cite{Wunder2015_ASILOMAR} propose a common overloaded control channel to simultaneously detect the activity and CIRs. To this end the $i$-th user sends a pilot signal $\vec p^{(i)} \in \CC^w$ into the control channel of bandwidth $w \in \NN$. The CIRs with respect to a single user $\x^{(i)} \in \CC^n$ can be stacked into a vector $\x \in \CC^{Nn}$ allowing for a maximum number of $N$ users. The received signal in an OFDM-type setting  has the measurement matrix \cite{Wunder2015_ASILOMAR,WunderEtAl:mMTC:asilomar:2017}
\begin{equation}\label{eq:CCRAMeasurement}
	\A = \Phi\mathcal{F}^{(w)} D(p).
\end{equation}
The matrix $D(p) \in \CC^{w \times Nn}$ depends on the pilot signals $p = [\vec p^{(1)}, \vec p^{(2)}, \ldots, \vec p^{(N)}]$.
The matrix $\Phi $ is again a mask selecting a support drawn uniformly from $[Nn]$.
In a network with a large number of users $N$ having only sporadic traffic, $\x$ will be block-sparse with a number $s$ of maximally simultaneously active users much smaller than $N$. Furthermore, one observes that the CIR for each user itself is $\sigma$-sparse, see Ref.\ \cite{BajwaEtAl:2010} and references therein. In this context, it is hence natural to assume that 
$\x$ is an $(s,\sigma)$-sparse signal as introduced above. 

Hierarchically structured problems also arise in bilinear compressed sensing problems that involve two sparse vectors \cite{WunderEtAl:2018}. 
In such settings the resulting vectorization of the lifted matrix is hierarchically sparse. 
The hierarchically sparse framework does not exploit the low-rank structure present in such bilinear problems and, thus, gives sub-optimal recovery guarantees in the number of required samples. Nonetheless, one can argue that the \HiHTP\ algorithm exploits as much structure as possible with an efficient, exact projection and is a viable low-complexity algorithm in this context \cite{WunderEtAl:2018,FoucartEtAl:2019}. 
Note that alternative approaches focus on approximate projections for the reconstruction of group sparse signals \cite{BahEtAl:2019}.

\subsection{Related work}
Sparse user activity in MTC as well as sparse channel estimation has meanwhile been studied in many works, see 
Ref.~\cite{Wunder2015_ACCESS}. Both sparsity features can be separately exploited with standard compressed sensing algorithms, e.g., $\ell_1$-norm optimization \cite{BajwaEtAl:2010, ZhuGiannakis:2011}.
The combined sparsity structure arising from sparse CIRs in a scheme with sporadic traffic
has been investigated recently by Ref.~\cite{SchepkerEtAl:2013} (and follow up work) where
a hierarchical version of the Orthogonal Matching Pursuit (OMP) algorithm was formulated
in the same context but without providing a proof of convergence.

This work follows the outline of model based compressed sensing \cite{BarCevDua10} and makes use of generalized RIP constants for unions of finite-dimensional linear subspaces as proposed in Ref.~\cite{BluDav09}. 
The structure of hierarchically sparse signals can be seen as a combination of  level sparsity \cite{AdcockEtal:2013,LiAdcock:2016} and block sparsity \cite{EldarMishali:2009a,EldarMishali:2009b}.  In the latter body of work it has been pointed out that block-sparse signals can be recovered by minimizing the mixed norm $\| \cdot \|_{\ell_2/\ell_1}$
which amounts to the sum of the $\ell_2$-norms of the blocks.
Also, corresponding block thresholding algorithm have been proposed, such as Group Orthorgonal Matching Pursuit \cite{MajumdarWard:2009}.  

The notion of hierarchical sparsity and the special case of block sparse vectors with sparse blocks has been introduced in Refs.\ \cite{SprechmannEtAl:2010,FriedmanEtAl:2010, SprechmannEtAl:2011, SimonEtAl:2013}. Herein a linear combination of the $\ell_1$-norm and the mixed $\ell_2/\ell_1$-norm is employed as a regularizer for $(s,\sigma)$-sparse signals. The resulting convex optimization problem reads as
\begin{equation}\label{eq:HiLassoProblem}
	\operatorname{minimize} \frac{1}{2} \norm{\y-\A\x}^2 + \mu \|\x\|_{\ell_1} + \lambda \| \x \|_{\ell_2/\ell_1}.
\end{equation}
The iterative soft thresholding algorithm solving this optimization problem was dubbed \HiLasso. For \HiLasso\ convergence and recovery guarantees have been shown based on generalized notions of coherence for the block-sparse structure. To the authors' knowledge there are no results available on the required sampling complexity for specific ensembles of measurement matrices such as Gaussian measurements to fulfill the coherence conditions.
A generalization of the Orthogonal Matching Pursuit algorithm to $(s,\sigma)$-sparse vectors has been
proposed in Ref.\ \cite{LiuSun:2011} providing only numerical indications of its performance. 
In Refs.\ \cite{BachEtAl:2012,JenattonEtAl:2011a,JenattonEtAl:2011b} generalizations of the regularizer of \eqref{eq:HiLassoProblem} have been constructed to reconstruct signals with more general sparsity structures that involve groupings of vector entries. 

The relation between the RIP-constants of a group of matrices  and the corresponding constant for their Kronecker product has been investigated in Refs.~\cite{JokarMehrmann:2009,DuarteBaraniuk:2012}. It was established that to build a matrix $\A_1 \otimes \dots \otimes \A_L$ with the $s$-RIP, we need each $\A_\ell$ to 
exhibit the $s$-RIP. The situation can be remedied if one assumes the signal to also feature a Kronecker product structure \cite{DuaBar12}. This structure is much more restrictive than the hierarchical sparsity in this work. There are analytical indications that such combined low-rank and sparsity structures can not be completely exploited in algorithmic solutions following standard compressed sensing approaches \cite{OymakFazelEldarHassibi:2015,Magdon-Ismail:2017}, see also Refs.~\cite{MuEtAl:2013,KlieschEtAl:2019}. 

\subsection{Notation}
\label{sec:setting}
The set of strictly positive integers being not larger than $n \in \ZZ^+$ is denoted by $[n]\coloneqq \{1,2,\dots, n\}$. 
By $\KK$ we denote a field that is either the real numbers $\RR$ or the complex numbers $\CC$. 
The imaginary unit is denoted by $\i$ so that $\i^2 = -1$. 
The $\ell_q$-norm of $\x \in \KK^d$ is denoted by 
$\norm{\x}_{\ell_{q}}\coloneqq(\sum_{j}|x_{j}|^{q})^{1/q}$ and the Euclidean norm by
$\norm{\x}\coloneqq\norm{\x}_{\ell_{2}}$. 
With 
$\supp(\x)\coloneqq\{j\,:\,x_{j}\neq0\}$ we denote the support of a vector 
$\x \in \KK^d$. 
Given a matrix $\A \in \KK^{m\times d}$ we refer by $\A_\Omega$ with $\Omega \subset [d]$ to the $m \times |\Omega|$ \emph{submatrix} consisting only of the columns indicated by $\Omega$. Analogously, for a vector $\x \in \KK^d$ the \emph{restriction} to $\Omega$ is denoted $\x_\Omega \in \KK^{|\Omega|}$. 
We write $\x\proj_\Omega$ for the \emph{projection} of $\x$ to the subspace of $\KK^{d}$ with support $\Omega$, i.e.
\begin{equation}\label{eq:def_proj}
  (\x\proj_\Omega)_k \coloneqq 
  \begin{cases}
    x_k & \text{ for } k \in \Omega,\\
    0 & \text{ for }k \notin \Omega.
  \end{cases}
\end{equation}
The complement of $\Omega \subset [d]$ is denoted by $\bar{\Omega}\coloneqq [d] \setminus \Omega$.

\section{The algorithm: \HiHTP}
Established thresholding and greedy compressed sensing algorithms, 
e.g.,  CoSaMP \cite{Needell08} and \emph{Hard Thresholding Pursuit (\htp)} 
\cite{Foucart:2011}, follow 
a common strategy: In each iteration, first, a proxy to the signal $\x$ is computed from the 
previous approximation to the signal and from the measurement vector $\y$. 
From this proxy a guess for the support of $\x$ is inferred by applying a thresholding operation. As a second step of the iteration, the best $\ell_2$-norm approximation to the measurements 
compatible with this support is calculated. 
In {\htp}, the \emph{$s$-sparse thresholding operator (TO)} 
$L_s: \CC^n \to \{ \Omega \subset [n] \mid |\Omega| = s \}$
is given by
 \begin{equation}\nonumber
	L_s(\z) \coloneqq \{\text{indices of $s$ largest entries of $\z$ in magnitude}\}. 
\end{equation}
This operator returns the support 
$L_s(\z) \subset [n]$ of the best $s$-sparse approximation to $\z$. Should there be several best $s$-sparse approximation (when several entries are equal in magnitude), we can arbitrarily choose any of them.

The basic idea of model-based compressed sensing \cite{BarCevDua10} is to adapt this TO to the model in order to improve the performance of the algorithm.
We denote the TO that yields the support of the best $(s,\sigma)$-sparse approximation to $\z \in \CC^{nN}$ by $L_{s, \sigma}$, with a similar handling of the case that there are several optimal approximations. Importantly, in this case, $L_{s, \sigma}(\z)$ can be easily calculated: 
We apply the $\sigma$-sparse TO to each block separately, 
select the $s$ active blocks as the largest truncated blocks in $\ell_2$-norm, and
collect the remaining $s\cdot \sigma$ indices in the set $L_{s, \sigma}(\z)$. 
This prescription is illustrated in Fig.~\ref{fig:HiThresholding}.

\begin{figure}[tb]
	\centering
	\input{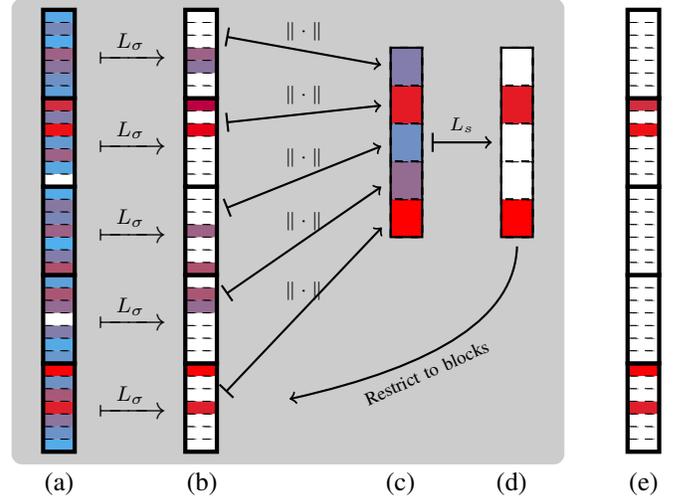}
	\vspace{-.6cm}
	\caption{In this figure, the evaluation of the hierarchical thresholding operator $L_{s,\sigma}$ is illustrated. Starting with a given dense vector (a), each block is thresholded to the best $\sigma$-sparse approximation (b). To determine the $s$ dominant blocks, the $\ell_2$-norm is calculated for each block. The resulting vector (c) of length $N$ is again thresholded to its best $s$-sparse approximation (d). The resulting blocks indicated by the $s$-sparse approximation (d) are selected from the $\sigma$-sparse approximation (b). The remaining $(s,\sigma)$-sparse support (e) is the output of $L_{s,\sigma}$.}
	\label{fig:HiThresholding}
\end{figure}

Using $L_{s,\sigma}$ instead of $L_{s}$ in the {\htp} yields Algorithm~\ref{alg:HiHTP}, which we call {\HiHTP}, as it is designed to recover a \emph{hi}erarchically structured sparsity.
\begin{algorithm}[tb] 
	\caption{(\HiHTP)} 
	\label{alg:HiHTP}
	\begin{algorithmic} [1]
 		\REQUIRE measurement matrix $\A$, measurement vector $\y$, block column sparsity $(s,\sigma)$
 		\STATE $\x^0 = 0$ 
 		\REPEAT
 			\STATE $\Omega^{k+1} = L_{s,\sigma} (\x^k + \A^\ast (\y - \A \x^k))$\label{alg:HiHTP:TH}
 			\STATE $\x^{k+1} = \argmin_{\vec z \in \CC^{Nn}} \{ \|\y - \A\vec z\|,\  \supp(\vec z) \subset \Omega^{k+1} \}$ \label{alg:HiHTP:LS}
 		\UNTIL stopping criterion is met at $\tilde{k} = k$
 		\ENSURE $(s,\sigma)$-sparse vector $\x^{\tilde{k}}$
	\end{algorithmic}
\end{algorithm}
As in the original {\htp} proposal, a natural stopping criterion is that two subsequent supports coincide, i.e.\ $\Omega^{k+1} = \Omega^k$.

A similar modification employing $L_{s,\sigma}$ can also be applied to other compressed sensing algorithms such as the Iterative Hard Thresholding~\cite{BlumensathDavies:2008}, as done in the companion work on mMIMO \cite{Wunder2019_TWC}, the Subspace Pursuit~\cite{DaiMilenkovic:2009} or Orthogonal Matching Pursuit, see e.g., Ref.\ \cite{Tropp:2004,LiuSun:2011} and references therein.

\paragraph*{Computational complexity}
The computational complexity of {\HiHTP} scales identically as for the original \htp. The algorithms only differ in the thresholding operators. 
Using a quick-select algorithm \cite{Hoare:1961} to perform the thresholding operator yields computational costs of $\mathcal{O}(Nn)$ for both algorithms. 
The overall performance is dominated by the costs of matrix vector multiplication with the measurement matrix $\A$ and $\A^\ast$, e.g., for
the calculation of the proxy. Therefore the computation time per iteration in general scales as $\mathcal{O}(mNn)$. If the measurement matrix allows for a fast matrix vector 
multiplication this scaling can be improved \cite{Needell08}.

\section{Analytical results} \label{sec:analytic}
For the analysis of hierarchically sparse recovery schemes, we use a special version of the general restricted isometry property (RIP) for unions of linear subspaces \cite{BluDav09}. In the following we formulate the \emph{hierarchical RIP (HiRIP)} for $(s,\sigma)$-sparse vectors. A generalized version of HiRIP for arbitrary hierarchical sparse vectors is given in Section~\ref{sec:generalHierachicalSparsity}. 

\begin{definition}[HiRIP]\label{def:RIP}
	Given a matrix $\A \in \KK^{m\times Nn}$, we denote by $\delta_{s,\sigma}$ the smallest $\delta \geq 0$ such that 
	\begin{equation}\label{eq:def:RIP}
	  (1-\delta)\|\x\|^2 \leq \|\A\x\|^2 \leq (1+\delta) \|\x\|^2 
	\end{equation}
	for all $(s,\sigma)$-sparse vectors $\x \in \KK^{nN}$.
\end{definition}

We will also make use of the standard $s$-sparse RIP constants $\delta_s$ fulfilling the condition \eqref{eq:def:RIP} for all $s$-sparse vectors. Since every $(s,\sigma)$-sparse vector is $s\cdot\sigma$-sparse, we immediately observe that $\delta_{s,\sigma} \leq \delta_{s \cdot \sigma}$.
Furthermore, the $(s,\sigma)$-sparse RIP constants are non-decreasing in both their indices, 
\begin{equation}\label{eq:delta-hierarchy}
\begin{aligned}
  \delta_{s,\sigma} \leq \delta_{s+1,\sigma} \, ,
  \quad
  \delta_{s,\sigma} \leq \delta_{s,\sigma+1}
\end{aligned}
\end{equation}
for all $s$ and $\sigma$.

The $s\cdot \sigma$-sparsity of $(s,\sigma)$-sparse vectors allows to apply standard compressed sensing algorithms for the recovery from linear measurements for which  recovery guarantees have been shown. Due to the their similarity, {\HiHTP} inherits the results for the success of  recovery for measurement matrices with small RIP constants that have been established for {\HTP} \cite{Foucart:2011}. See also Ref.\ \cite[Chapter 6.3]{FouRau13}. 
The RIP condition on the measurement matrix derived therein is $\delta_{3s} < 1/\sqrt{3}$.

One main technical insight of model-based compressed sensing \cite{BarCevDua10} is that the generalized RIP of 
Ref.\ \cite{BluDav09} allows for the same proof techniques as the standard RIP \cite{CanTao05} and leads to improved recovery guarantees. 
Following this strategy, we can establish less restrictive RIP-type conditions for successful recovery in terms of the custom tailored $(s,\sigma)$-sparse HiRIP constants. The resulting statement is the following.

\begin{theorem}[Recovery guarantee]\label{thm:recGarant}
	Suppose that the following HiRIP condition holds
	\begin{equation}\label{eq:recGarant:RIP}
		\delta_{3s,2\sigma} < \frac{1}{\sqrt{3}}.
	\end{equation}
	Then, for $\x \in \CC^{Nn}$, $\ev \in \CC^m$, and $\Omega \subset [N]\times [n]$ an $(s,\sigma)$-sparse support set, the sequence $(x^k)$ defined by {\HiHTP} (Algorithm~\ref{alg:HiHTP}) with $\y = \A\x\proj_\Omega + \ev$ satisfies, for any $k\geq 0$, 
	\begin{equation}
		\norm{\x^k - \x\proj_\Omega} \leq \rho^k \norm{\x^0 - \x\proj_\Omega} + \tau \norm{\ev},
	\end{equation}
	where 
	\begin{equation}
		\rho =  \left({\frac{2 \delta_{3s,2\sigma}}{1-\delta_{2s,2\sigma}^2}}\right)^{1/2} < 1  
	\end{equation}
	and $\tau \leq 5.15/(1-\rho)$.
\end{theorem}
The complete proof is given in Appendix~\ref{app:proofs}. The proof proceeds verbatim along the lines of the proof of the convergence result for {\HTP} \cite{Foucart:2011, FouRau13}. The modified thresholding operator $L_{s,\sigma}$ can be treated analogously to standard thresholding operator in the original proof. The main difference is that while the proof for {\HTP} uses standard RIP constants to bound the deviation of the algorithm's output $\x^k$ from the original signal $\x_\Omega$, we are in a position to employ the HiRIP constants in these bounds for {\HiHTP}. The crucial observation regarding the HiRIP constants is the following. 
\begin{observation}[Support unions]\label{observation:supportUnions}
	For $i=1,2$ let $\Omega_i \subset [N]\times [n]$ be an $(s_i,\sigma_i)$-sparse support and $\A\in \KK^{m\times Nn}$
	with HiRIP constants $\delta_{s, \sigma}$. Then
	\begin{equation}\label{eq:supportUnions}
		\norm{\Id- (\A_{\Omega_1 \cup \Omega_2})^\ast \A_{\Omega_1 \cup  \Omega_2}}
		\leq 
		\delta_{s_1+s_2,\sigma_1+\sigma_2}.
	\end{equation}
\end{observation}
The statement follows directly from hierarchy \eqref{eq:delta-hierarchy}  of the HiRIP constants and the observation that the union $\Omega_1 \cup \Omega_2$ has at most $s_1+s_2$ blocks and each block is at most $\sigma_1+\sigma_2$ sparse. 

In fact, one can prove an alternative bound replacing the right hand side of \eqref{eq:supportUnions} by $3 \max\{\delta_{s_{\mathrm{max}},\sigma_1+\sigma_2}, \delta_{s_1+s_2, \sigma_{\mathrm{max}}}\}$, where $s_{\mathrm{max}} \coloneqq \max\{s_1,s_2\}$ and $\sigma_{\mathrm{max}} \coloneqq \max\{\sigma_1,\sigma_2\}$.  The latter bound uses smaller HiRIP constants compared to Observation~\ref{observation:supportUnions}. However, for Gaussian measurements the formulation of Observation~\ref{observation:supportUnions} will lead to slightly smaller constants. 
The observation allows us to bound terms involving the sum of two and more $(s,\sigma)$-sparse vectors. Employing this in the proof of Theorem~\ref{thm:recGarant} yields the modified RIP condition~\eqref{eq:recGarant:RIP}. 

The trivial bound $\delta_{2s,2\sigma} \leq \delta_{4s\sigma}$ does not indicate an improvement by the theorem compared to the established bound $\delta_{3s} < 1/\sqrt{3}$. 
But the decreased number of subspaces which contribute to $(s,\sigma)$-sparse vectors compared to the set of all $s\sigma$-sparse vectors allows us to provide tighter bounds for $\delta_{s,\sigma}$ compared to $\delta_{s\sigma}$ for specific random matrices.

\paragraph{Gaussian measurements}
Building on Ref.~\cite{BarCevDua10,BluDav09}, we establish the following result for the case of real Gaussian measurement matrices.
\begin{theorem}[HiRIP for real Gaussian measurements]\label{thm:gaussianBscRIP}
	Let $\A$ be an $m \times (N\cdot n)$ real matrix with {i.i.d.} Gaussian entries and $m < Nn$. 
	For $\epsilon, \delta > 0$, assume that 
	\begin{equation}
		m \geq \frac{2}{c \delta^2}\left( s \ln \left( \frac{\e N}{s}\right) + s\sigma \ln\left(\frac{ \e n}{\sigma }\cdot \frac{ 12 }{ \delta}\right) +\ln(\epsilon^{-1}) \right),
	\end{equation}
	with $c>0$ a universal constant.
	Then, with probability of at least $1 - \epsilon$, the restricted isometry constant $\delta_{s,\sigma}$ of $\A/\sqrt{m}$ satisfies
	\begin{align}
		\delta_{s,\sigma} \leq \delta. 
	\end{align}
\end{theorem}

\begin{proof}[Proof]
The bounds on RIP constants for typical random matrices, {e.g.,} Gaussian matrices or matrices with sub-Gaussian rows, proceed in two steps, see e.g., Ref.\ \cite[Chapter~9]{FouRau13}. From the properties of the specific random measurement matrix they derive a bound on the probability that $\norm{\A \x}^2 - \norm{\x}^2 \leq \epsilon \norm{\x}^2$ for a fixed $\x$. With such a bound the RIP constant can be upper bounded by taking the union bound over all relevant subspaces that might contain $\x$. For example, for the standard RIP constant $\delta_s$ for $s$-sparse vectors in $\RR^n$ one has to consider the $\binom{n}{s}$ subsets of $[n]$ with cardinality $s$.  In this way, one establishes a bound on $\delta_s$ to hold with high probability provided that the number of samples of the measurement matrix is lower bounded, $m \geq \tilde m$, with $\tilde m \in \mathcal{O}\bigl(s\ln (n/s)\bigr)$.

Theorem~\ref{thm:gaussianBscRIP} is a direct corollary of Theorem~1 Ref.\ \cite{BarCevDua10} applied to $(s,\sigma)$-sparse vectors. Said theorem establishes that an {i.i.d.} Gaussian matrix $\A \in \RR^{m\times d}$ has the RIP property with RIP constant upper bounded by $\delta$ with respect to vectors of a union of $L$ $k$-dimensional subspaces with probability of $1-\epsilon$ provided that 
	\begin{equation}\label{eq:bludav:samplingcomplexity}
		m \geq \frac{2}{c \delta^2} \left( \ln(2L) + k\ln\left(\frac{12}{\delta}\right) + \ln{\epsilon^{-1}} \right).
	\end{equation}
	An $(s, \sigma)$-sparse signal is in the union of \begin{equation}
	\binom{N}{s}\binom{n}{\sigma}^{\sigma} \leq \left(\frac{eN}{s}\right)^s \left(\frac{en}{\sigma}\right)^{s\sigma} 
		\end{equation}
	$s\sigma$-dimensional subspaces.
	Thus, \eqref{eq:bludav:samplingcomplexity} yields the statement of Theorem~\ref{thm:gaussianBscRIP} for the case of $(s, \sigma)$-sparse vectors. 
\end{proof}

Combining Theorem~\ref{thm:recGarant} and Theorem~\ref{thm:gaussianBscRIP}, we establish the successful recovery of $(s,\sigma)$-sparse vectors from Gaussian measurements provided that the number of samples $m \geq \tilde{m}$, where the lower bound on the number of samples $\tilde m$ parametrically scales like
\begin{equation}
	\tilde{m} \in \mathcal{O} \bigl(s \ln ({N}/{s}) + s \sigma \ln ({n}/{\sigma})\bigr).
\end{equation}
For Gaussian measurements the standard RIP analysis of {\HTP} in our setting yields recovery guarantees for a number of samples $m \geq \tilde{m}$ with parametric scaling
	$\tilde{m} \in \mathcal{O}\bigl(s\sigma \ln ({Nn}/({s\sigma}))\bigr) =\mathcal{O}\bigl(s\sigma \ln ({N}/{s}) +s\sigma \ln ({n}/({\sigma})) \bigr)$,
see, e.g., Ref.\ \cite{FouRau13}. 
We observe that in the regime of $N \gg n$ (i.e. the number of blocks is much larger than the size of each block), the required number of samples may be significantly lower. 

For the direct application (see Section~\ref{subsec:application}) similar improved bounds for subsampled Fourier matrices are required. However, the different proof techniques render the adaption of such bounds to hierarchical sparse signals more difficult. We leave analytic results for the recovery of hierarchical sparse signals from partial Fourier measurements to future work. On the other hand, in Section~\ref{sec:numerics} we find
numerically similar results for Gaussian measurements and partial Fourier measurements. Moreover, in the next section
we find a HiRIP characterization for Kronecker product measurements which is important for the applications.
\paragraph{Kronecker product measurements}

A second class of measurement matrices that feature the HiRIP are Kronecker products of matrices satisfying the RIP. The result is given in the following theorem:
\begin{theorem}\label{thm:hirip}
  Given $\A \in \KK^{M\times N}$ having $s$-sparse RIP constant $\delta^\A_s$ and $\B \in \KK^{m\times n}$ with $\sigma$-sparse RIP constant $\delta^\B_\sigma$,  
  \begin{equation}
    \A \otimes \B : \KK^{Nn} \to \KK^{Mm}
  \end{equation}
  has the $(s,\sigma)$-sparse HiRIP constant $\delta_{s,\sigma}$ satisfying
  \begin{equation}
     \delta_{s,\sigma} \leq \delta^\A_s + \delta^\B_\sigma + \delta^\A_s\delta^\B_\sigma.
   \end{equation} 
\end{theorem}

Before presenting the proof of Theorem \ref{thm:hirip}, we need to introduce some 
notation. First, we denote by $\vecmap: \KK^{N\times n} \to \KK^{Nn}$ the canonical isomorphism of column-wise vectorization. In other words, $\vecmap$ is defined by linear extension of the requirement $\vecmap(\vec E_{i,j}) = \vec{e}_i \otimes \vec{e}_j$, where $\vec E_{i,j} = \vec{e}_i \vec{e}_j^T \in \KK^{N\times n}$ denotes the matrix with only one non-vanishing unit entry in the $i$-th row and $j$-th column. The Kronecker product of matrices $\A$ and $\B$ is always understood as 
\begin{equation}
  \A \otimes \B = \begin{pmatrix} 
    a_{1,1} \B& \ldots & a_{1,N}\B  \\
    \vdots & \ddots & \vdots \\
    a_{m,1} \B & \ldots & a_{m,N}\B
  \end{pmatrix}.
\end{equation}
This convention justifies the term \emph{column-wise vectorization}. 

It will be convenient to also implicitly make use of \emph{row-wise vectorization}, which can be defined as $\vec X \mapsto \vecmap(\vec X^T)$.  Passing from one vector representation to the other amounts to applying the \emph{flip operator} $F_{N,n}: \KK^{Nn} \to \KK^{Nn}$. The latter is defined as the linear extension of the mapping $\vec{e}_i \otimes \vec{e}_j \mapsto \vec{e}_j \otimes \vec{e}_i$.
The action it induces on linear operators $\A\otimes \B$ that are defined through Kronecker products is the same as switching the order of the matrices involved. To be precise:
\begin{lemma}\label{lem:flipping} For $\A \in \KK^{M\times N}$ and $\B \in \KK^{m\times n}$ and $\vec X \in \KK^{N\times n}$ it holds that 
\begin{equation*}
  \A \otimes \B= F_{m,M} (\B \otimes \A) F_{N,n} 
\end{equation*}
and
\begin{equation*}\label{eq:fliptransp}
  F_{N,n} \vecmap(\vec X) = \vecmap(\vec X^T). 
\end{equation*}
\end{lemma}

We now have all the tools we need to prove Theorem \ref{thm:hirip}.

\begin{proof}[Proof of Theorem~\ref{thm:hirip}]
Let $\x \in \KK^{Nn}$ be hierarchically $(s,\sigma)$-sparse. With the help of Lemma~\ref{lem:flipping} we find
\begin{align*}
  \norm{(\A \otimes \B)\x}^2  &= \norm{(\A \otimes \Id_n)(\Id_N \otimes \B) \x}^2 \\
  &=\norm{F_{m,M} (\Id_n \otimes \A) F_{N,n} (\Id_N \otimes \B) \x}^2 \\
  &= \norm{(\Id_n \otimes \A) F_{N,n} (\Id_N \otimes \B) \x}^2,
\end{align*}
where the last line follows from the fact that $F_{m,M}$ is unitary.  The vector $(\Id_N \otimes \B)\x$ has only non-vanishing entries in $s$ of its $N$ blocks. Therefore, the flipped vector $\vec{h} \coloneqq F_{N,n}(\Id_N \otimes \B)\x$ consists of blocks $\vec{h}_i \in \KK^N$ with $i \in [n]$ that are at most $s$-sparse each. This allows us to apply the $s$-sparse RIP property of $\A$ for each of the blocks
\begin{align} \label{eq:crucialStep1}
  &\norm{(\Id_n \otimes \A) \vec{h}}^2 = \sum_{i\in [n]} \norm{\A \vec{h}_i}^2  \leq (1+ \delta_s) \norm{\vec{h}}^2.
\end{align}
Making use of the unitarity of the flip once again, the $\ell_2$-norm of $\vec{h}$ is identical to 
\begin{align} \label{eq:crucialStep}
  \norm{\vec{h}}^2 &= \norm{(\Id_N \otimes \B)\x}  = \sum_{i \in [N]} \norm{\B\x_i},
\end{align}
where $\x_i \in \KK^n$ $i \in [N]$ are the $\sigma$-sparse blocks of $\x$. Every term of the sum is bounded by the $\sigma$-sparse RIP of $\B$ yielding 
\begin{equation*}
  \norm{\vec{h}}^2 \leq (1+\delta_\sigma)\norm{\x}^2. 
\end{equation*}
In summary, we have established
\begin{equation*}
  \norm{(\A\otimes \B)\x}^2 \leq (1+\delta_s) (1+ \delta_\sigma) \norm{\x}.
\end{equation*}
The lower RIP bound can be derived in the same way, completing the proof. 
\end{proof}

The main consequence of Theorem~\ref{thm:hirip} is that it allows to construct a new class of measurement matrices for which the \HiHTP\ algorithm is guaranteed to succeed. More precisely, we get the following corollary. 

\begin{corollary}\label{cor:recovery}
  Let $\A \in \KK^{M\times N}$ and $\B \in \KK^{m\times n}$, and suppose that the following RIP-conditions hold
  \begin{align*}
    \delta^{\A}_{3s}, \delta^{\B}_{2\sigma} \leq \sqrt{\frac{\sqrt{3}+1}{\sqrt{3}}} -1. 
  \end{align*}
  Then, for $\x \in \KK^{nN}$, $\vec e \in \KK^{Mm}$ and $\Omega \subseteq [N] \times [n]$ and $(s,\sigma)$-sparse support set, the sequence $\x^k$ defined by the \HiHTP\ Algorithm~\ref{alg:HiHTP} with $y = (\A \otimes \B) \x\proj_\Omega +\ev$ satisfies, for any $k\geq 0$
  \begin{align*}
    \norm{ \x^k - \x \proj_\Omega} \leq \rho^k \norm{\x^0- \x \proj_\Omega} + \tau \norm{\ev},
\end{align*}   
where
\begin{align*}
  \rho = \left(\frac{2(\delta^{\A}_{3s}+ \delta^{\B}_{2\sigma} + \delta^{\A}_{3s} \delta^{\B}_{2\sigma})} {1 -(\delta^{\A}_{3s}+ \delta^{\B}_{2\sigma} + \delta^{\A}_{3s} \delta^{\B}_{2\sigma})^2 }\right) <1.
\end{align*}
\end{corollary}
\begin{proof}
  We simply need to note that Theorem \ref{thm:hirip} implies that $$\delta_{3s,2\sigma}^{\A \otimes \B} \leq (\delta^{\A}_{3s}+ \delta^{\B}_{2\sigma} + \delta^{\A}_{3s} \delta^{\B}_{2\sigma} )< \frac{1}{\sqrt{3}}.$$
  The rest follows from Theorem~\ref{thm:recGarant}.
\end{proof}

Taking any pair of random matrices $\A \in \KK^{M\times N}$ and $\B \in \KK^{m\times n}$ both guaranteed to possess the $s$- and $\sigma$-RIP with high probability, respectively, $\A \otimes \B$ will satisfy the $(s,\sigma)$-HiRIP with high probability as well. As an example, we can use Gaussian random matrices with shapes $M \gtrsim s \log N $ and $m \gtrsim \sigma \log n $, resulting in a measurement matrix $\vec A \otimes \vec B \in \KK^{\mathfrak{m}\times nN}$ where $$\mathfrak{m} \gtrsim s\sigma \log N\, \log n .$$
Hence, a measurement scheme using Kronecker matrices will need slightly more measurements than the fully Gaussian matrices to have the HiRIP. 
This price could, however, sometimes be worth paying. First, we get a vast reduction in space needed to store the matrix: $MN + mn$ instead of $MN\cdot mn$. Also, as has been discussed in the introduction, there are applications where the Kronecker structure of a measurement process is inherent. 

\newcommand{\block}{\mathrm{block}}
\newcommand{\abs}[1]{\left\vert #1 \right\vert}
\newcommand{\set}[1]{\left\lbrace #1\right\rbrace}
\newcommand{\sse}{\subseteq}
\newcommand{\sprod}[1]{\left\langle #1 \right\rangle}
\newcommand{\geqsim}{\gtrsim}
\newcommand{\leqsim}{\lesssim}

\paragraph*{Relations to hierarchical coherence measures}

In Ref.~\cite{SprechmannEtAl:2011}, the authors introduce the {\HiLasso} algorithm, and conduct a theoretical analysis based on different generalized coherence measures. In the following, we will connect two of these measures with the RIP-based analyses.
Let us denote the column of $\vec A \in \KK^{m \times Nn}$ acting on the $j$-th entry of the $i$-th block by $\vec{a}_{i,j}$. The \emph{subcoherence} of $\vec{A}$ is then defined as 
\begin{equation*}
	\nu(\vec{A}) = \sup_{i \in [N]} \sup_{\substack{k,j \in [n] \\ k \neq j}} \abs{\sprod{\vec{a}_{i,k}, \vec{a}_{i,j}}}.
\end{equation*}
Intuitively, $\nu$ measures the maximal mutual coherence within a single block.
To also account for the coherence between the blocks, we consider the following notion of 
\emph{sparse block-coherence}
\begin{align*}
	\mu_{\mathrm{block}}^{\sigma \sigma }(\vec{A}) &= 
	\sup_{\substack{i,j\in [N] \\i\neq j}}
	\tfrac{1}{n}\rho^{\sigma\sigma}(\vec{A}_{\Omega_i}^*\vec{A}_{\Omega_j}).
\end{align*}
Here, $\Omega_i$ is the entire $i$-th block of indices, $\Omega_i = \set{i}\times [n]$ and  $\rho^{\sigma\sigma}$ is the  \emph{$\sigma$-sparse singular value},
\begin{equation*}
	\rho^{\sigma\sigma}(\vec{A}_{\Omega_i}^*\vec{A}_{\Omega_j}) 
	= \sup_{\substack{
		\norm{\vec{u}}, \norm{\vec{v}}\leq 1  
		\\ 
		\text{$\vec{u}$, $\vec{v}$ $(1,\sigma)$-sparse} \\
		\supp \vec{u} \sse \Omega_i, \supp \vec{v} \sse \Omega_j}}\sprod{\vec{u},\vec{A}_{\Omega_i}^*\vec{A}_{\Omega_j}\vec{v}}.
\end{equation*}

Coherence measures are typically coarser notions for ensuring successful recovery compared to RIP. For example, for matrices with normalized columns, the standard coherence 
\begin{equation*}
	\mu(\A) = \sup_{\substack{(i,i') \neq (j,j') \in [N]\times [n]}} \abs{\langle \vec a_{i,i'},\vec a_{j,j'} \rangle}
\end{equation*}
dominates the standard $s$-sparse RIP constant $\delta_s$ as \cite[Theorem 5.3]{FouRau13}
\begin{equation}\label{eq:standardCoherenceRelation}
	\delta_s(\A)\leq (s - 1) \mu(\A).
\end{equation}
For the HiRIP-constants we can derive a similar relation using  
the subcoherence $\nu$ and the sparse block-coherence $\mu^{\sigma \sigma}_{\mathrm{block}}$.
\begin{theorem}\label{thm:relationToCoherence}
Assume that $\vec{A} \in \KK^{m\times Nn}$ has normalized columns $\vec a_i$, i.e. $\norm{\vec{a}_{i}}=1$ for all $i \in [Nn]$. Then the HiRIP constant of $\vec A$ satisfies
\begin{align} \label{eq:RIPvsCoherence}
    \delta_{s, \sigma} \leq (s-1) \cdot n\mu_{\block}^{\sigma \sigma}(\A) + (\sigma-1) \cdot \nu(\A) .
\end{align}
\end{theorem}

Note the relation $\mu_{\block}^{\sigma \sigma} \leq  \frac{\sigma}{n} \mu$. 
Hence, the unusual scaling in $n$ here is merely due to the definition of $\mu_{\block}^{\sigma\sigma}$.
In fact, for $n=1$ (i.e. when there is no hierarchal structure), $\mu_\block^{\sigma\sigma}(\vec{A}) \leq \sigma \mu(\vec{A})$ and $\nu(\A)=0$. 
Hence, Theorem~\ref{thm:relationToCoherence} reproduces the well-known result \eqref{eq:standardCoherenceRelation} in this case.

\begin{proof}[Proof of Theorem~\ref{thm:relationToCoherence}]
    Let $\vec{x}$ be an $(s,\sigma)$-sparse vector. 
    Interpreting $\vec x$ as a member of $\KK^{N \times n}$, we may write
    \begin{align*}
        \vec{x} = \sum_{i \in S} \vec{e}_i \otimes \vec{g}_i
    \end{align*}
    for a set $S \sse [N]$ with $\abs{S}\leq s$, the standard unit vectors $\vec{e}_i \in \KK^N$, and $\sigma$-sparse vectors $\vec{g}_i \in \KK^n$. Consequently,
    \begin{align*}
        &\norm{\vec{A}\vec{x}}^2 
        = \sprod{\vec{A}\vec{x},\vec{A}\vec{x}} = \sum_{i,k \in S} \sprod{\vec{A}(\vec{e}_i \otimes \vec{g}_i),\vec{A}(\vec{e}_k \otimes \vec{g}_k)} \\
        &\quad=\sum_{i \in S}  \norm{\vec{A}(\vec{e}_i \otimes \vec{g}_i)}^2 + \sum_{i\neq k \in S}  \sprod{\vec{A}(\vec{e}_i \otimes \vec{g}_i),\vec{A}(\vec{e}_k \otimes \vec{g}_k)}.
    \end{align*}
We now derive bounds for both sums individually. 
To this end, let us define the operator $\vec A_i: \KK^n  \to \KK^m$, $\vec{g} \mapsto \vec{A}(\vec{e}_i \otimes \vec{g})$ for the columns of $\vec A$ corresponding to the $i$-th block. 
Then, for the first summand 
using \eqref{eq:standardCoherenceRelation} yields
\begin{align*}
	&\abs{\norm{\vec{A}(\vec{e}_i \otimes \vec{g}_i)}^2-\norm{\vec{g}_i}^2} 
	\leq \delta_\sigma(\vec{A}_i) \norm{\vec{g}_i}^2 \\
	&\qquad\leq (\sigma-1) \mu(\vec{A}_i) \norm{\vec{g}_i}^2  
	\leq (\sigma-1)\cdot \nu(\A) \norm{\vec{g}_i}^2
\end{align*}
Taking the sum over $i$ gives
\begin{align*}
	\abs{\sum_{i \in S}  \norm{\vec{A}(\vec{e}_i \otimes \vec{g}_i)}^2  - \norm{\x}^2} \leq (\sigma-1)\nu(\A)\norm{\x}^2,
\end{align*}
since $\norm{\x}^2 = \sum_{i \in S} \norm{\vec{g}_i}^2$.

We now move on to the second term involving the sum of cross-block terms.
Letting $S_i=\set{i} \times \mathrm{supp}(\vec{g}_i)$ for each $i$ and defining $\vec{C}:(\KK^\sigma)^s\to (\KK^{\sigma})^s$ through
\begin{align*}
	\vec{C}\left( \sum_{i\in S} \vec{e}_i \otimes \vec{g}_i\right) =  \left(\sum_{\substack{k\in S\\ k\neq i}} \vec{A}_{S_i}^*\vec{A}_{S_k}(\vec{e}_k \otimes \vec{g}_k)\right)_{i \in S},
\end{align*}
we see that the term can be written as $\sprod{\vec{x},\vec{C} \vec{x}}$. Since $\vec{C}$ is Hermitian, we have $\abs{\sprod{\vec{x},\vec{C} \vec{x}}} \leq \abs{\lambda_{\mathrm{max}}(\vec{C})} \norm{\vec{x}}^2$, where $\abs{\lambda_{\mathrm{max}}(\vec{C})}$ is the largest eigenvalue of $\vec{C}$ in magnitude. If however $\sum_{i=1}^s \vec{e}_i \otimes \vec{v}_i$ is a corresponding eigenvector, we have for each $i$
    \begin{align*}
     \abs{ \lambda_{\mathrm{max}}(\vec{A})} \norm{\vec{e}_i \otimes \vec{v}_i}^2 &=\abs{ \sum_{\substack{k\in S\\ k\neq i}} \sprod{\vec{e}_i \otimes \vec{v}_i, \vec{A}_{S_i}^*\vec{A}_{S_k}(\vec{e}_k \otimes \vec{v}_k)}} \\ 
      &\leq (s-1) \cdot n\, \mu_{\mathrm{block}}^{\sigma\sigma}(\A) \sup_{k} \norm{\vec{e}_k \otimes \vec{v}_k}^2.
    \end{align*}
    We used that $\vec{e}_k \otimes \vec{v}_k$ and $\vec{e}_i\otimes \vec{v}_i$ are $(1, \sigma)$-sparse supported on the disjoint blocks $\Omega_i$ and $\Omega_k$, respectively. Letting $i$ be such that $\norm{\vec{e}_i \otimes \vec{v}_i}^2$ is maximal, the above inequality implies that $\abs{\lambda_{\mathrm{max}}(\vec{C})}\leq (s-1)n\mu^{\sigma\sigma}_{\mathrm{block}}(\A)$.  We hence have 
   \begin{align*}
   		\abs{\sum_{i\neq k \in S}  \sprod{\vec{A}(\vec{e}_i \otimes \vec{g}_i),\vec{A}(\vec{e}_k \otimes \vec{g}_k) }}\leq (s-1)\cdot n\, \mu_{\text{block}}^{\sigma\sigma}(\vec{A})\norm{\vec{x}}^2.
   \end{align*}
   Combining the inequalities derived for both terms, we arrive at the theorem's statement.
\end{proof}
Theorem~\ref{thm:hirip} stated that measurement matrices that have a Kronecker product structure inherit a HiRIP from the RIP of the constituent matrices. 
Similarly, the generalized coherence measures may be bounded by the mutual coherences of $\vec{A}$ and $\vec{B}$.
\begin{theorem}\label{thm:KroneckerCoherence}
Let $\vec A \in \KK^{M \times N}$ and $\vec B \in \KK^{m \times n}$ with normalized columns $\vec a_i$ and $\vec b_i$, respectively, i.e.\ $\norm{\vec{a}_i}=\norm{\vec{b}_j}=1$ for each $i\in [N]$ , $j \in [n]$.
For the matrix $\vec{A} \otimes \vec{B}$, we have
\begin{align*}
\nu(\vec{A}\otimes \vec{B}) &= \mu(\vec{B}) \\
     \mu_\block^{\sigma\sigma} (\vec{A}\otimes \vec{B}) &\leq \frac{1}{n}\mu(\vec{A}) ( 1 + (\sigma-1) \cdot \mu(\vec{B})).
\end{align*}
\end{theorem}
The above inequalities together with \eqref{eq:RIPvsCoherence} imply that for a Kronecker measurement matrix
\begin{align*}
    \delta_s(\vec{A} \otimes \vec{B}) \leq & (\sigma-1)\cdot \mu(\vec{B}) + (s-1)\cdot\mu(\vec{A}) \\
    &+ (s-1)\cdot(\sigma-1)\cdot \mu(\vec{A})\mu(\vec{B}).
\end{align*}
Again, the mutual coherence conditions for the standard $s$-RIP on the local sparsity level, {i.e.} $(s-1)\mu(\vec{A}), (\sigma-1) \mu(\vec{B}) \leq C \delta$ leads to a $(s,\sigma)$-sparse HiRIP.

\begin{proof}[Proof of Theorem~\ref{thm:KroneckerCoherence}]
We begin by inspecting the subcoherence. We have for each $i\in [N]$ and $j, k \in [n]$, $j \neq k$
\begin{align*}
     \abs{\sprod{(\vec{A}\otimes \vec{B})_{i,k}, (\vec{A}\otimes \vec{B})_{i,j}}} &  = \abs{\sprod{\vec{a}_i,\vec{a}_i} \sprod{\vec{b}_k, \vec{b}_j}} =  \abs{\sprod{\vec{b}_k, \vec{b}_j}}.
\end{align*}
Maximizing over $j \neq k$ yields the claimed equality.
 
To estimate 
 $\mu_\block^{\sigma\sigma}$, consider arbitrary indices $i \neq j$ and $\sigma$-sparse and normalized vectors $\vec{x}_i$, $\vec{x}_j$. We then have
 \begin{align*}
   &\abs{\sprod{\vec{A}_{\Omega_i}(\vec{e}_i\otimes\vec{x}_i), \vec{A}_{\Omega_j}(\vec{e}_j \otimes \vec{x}_j)}} =  \abs{\sprod{\vec{a}_i, \vec{a}_j}} \abs{\sprod{\B \vec{x}_i , \B\vec{x}_j}} \\
   & \quad \quad \leq \mu(\vec{A})  \abs{\sprod{\B \vec{x}_i , \B\vec{x}_j}} \leq \mu(\vec{A})  \sup_{\substack{\norm{\vec{x}} \leq 1 \\ \vec{x} \text{ $\sigma$-sparse}}}  \norm{\B \vec{x}}^2 \\
   & \quad \quad \leq \mu(\vec{A})( 1+ \delta_\sigma(\vec{B}))  \leq \mu(\vec{A})(1 + (\sigma-1)\cdot \mu(\vec{B}))
 \end{align*}
 In the last estimate, we again applied \eqref{eq:standardCoherenceRelation}.  
\end{proof}

\section{Numerical results}
\label{sec:numerics}

In this section we compare the performance of {\HiHTP}, {\htp} and the {\HiLasso}-algorithm in numerical experiments for Gaussian and Fourier measurements.
All algorithms have been implemented in Matlab \cite{matlab}. For the implementation of {\HiLasso} the convex optimization problem \eqref{eq:HiLassoProblem} was directly solved using CVX \cite{cvx} with MOSEK \cite{MOSEK} to avoid  ambiguities in the implementation of the 
algorithm. Before the {\htp} and {\HiHTP} algorithms are applied the columns of the measurement matrix are normalised in $\ell_2$-norm. The entries of the result of the algorithms are subsequently multiplied by the normalising factors to restore the $\ell_2$-norms of the columns of the actual measurement matrix.  

\paragraph*{Block recovery rates}
In CS a common performance measure is the fraction of recovered signals given a certain amount of samples. In the context of block-structured signals there is a different measure of performance available, which is well motivated in multiple applications, e.g., in OFDM. We have consider a \emph{block} recovered if the reconstructed part of the signal deviates from the original signal by less than $\epsilon$ in $\ell_2$-norm. The choice of $\epsilon$ depends on the specific application and, in particular, its noise model. For a given number of samples we count the total number of recovered blocks. Furthermore, we distinguish between the number of recovered blocks which are non-zero and of those which are zero in the original signal. 

\begin{figure}[tb]
	\centering
	\input{plotGaussianComplete}
	\hspace{-.2cm}
	\input{plotGaussianBlocks}
	\vspace{-.6cm} 
	\caption{Left: Number of recovered signals from $100$ noiseless Gaussian samples over the number of measurements $m$ for {\HTP}, {\HiLasso} and \HiHTP,
	the latter introduced here. The signals consist of $N=30$ blocks of size $n=100$ with $s=4$ blocks having $\sigma=20$ non-vanishing real entries. Right: Number of recovered blocks over the number of measurements $m$ for {\htp} and \HiHTP. 
	  The dashed and dotted lines indicate the average number of correctly recovered zero and non-zero blocks, respectively. 
	  The solid lines show the total average number of recovered blocks. The signals consist of $N=30$ blocks with $s=4$ blocks having non-vanishing real entries.
	\label{fig:samplingComplexityGauss} }
\vspace{-.5cm}
\end{figure}

\paragraph{Gaussian measurements}

We consider $(s=4,\sigma=20)$-sparse signals consisting of $N=30$ blocks of dimension $n=100$. For each instance of the signal the 
supports are randomly drawn from a uniform distribution and the entries are i.i.d.\ real numbers  from a standard normal distribution. 
We subsequently run {\HTP}, {\HiHTP} and {\HiLasso} on 
noiseless Gaussian measurements of each signal.  For different numbers of measurements, we count the number of successfully recovered signals out of $100$ runs.  A signal is successfully recovered if it deviates by less than $10^{-5}$ in $\ell_2$-norm from the original signal. We observed that deviations are typically either significantly smaller or significantly larger than $10^{-5}$.

While it is straightforward to inform the {\HTP} and {\HiHTP} algorithm about the sparsity of the signal, the {\HiLasso} is calibrated by adjusting the weights $\mu$ and $\lambda$ in front of the regularizer terms in \eqref{eq:HiLassoProblem}. We have found that finding appropriate values for both weights requires extensive effort, especially in the presence of additive noise. In applications where the sparsity levels are approximately known hard-thresholding algorithms do not require additional calibration.
For noiseless Gaussian measurements, we have found numerically that {\HiLasso} yields good recovery rates for $\mu=0.4$ and $\lambda=0.5$ in our setting, see Appendix~\ref{app:calibration}. Therefore these parameters are used for {\HiLasso} in the tests. The results of all three algorithm are shown in Fig.~\ref{fig:samplingComplexityGauss} on the left.

It is often not required to reconstruct the entire signal in an application. Instead, a relevant measure of performance is the number of successfully recovered blocks.  In the following, a block is  successfully recovered if it deviates by less than $10^{-5}$ in Euclidean norm from the corresponding block of the original signal. For each number of measurements we average the number of recovered blocks over $100$ runs. 

Fig.~\ref{fig:samplingComplexityGauss} shows the resulting recovery rates on the right. While for {\HTP} the number of recovered blocks quickly decays 
for small numbers of samples, {\HiHTP} performs significantly better in this regime. Note that the minimal number of not recovered blocks is lower bounded by $2s$, which follows directly from the definition of  \HiHTP.  
Furthermore, the {\HiHTP} recovers the content of the active blocks accurately using less measurements then \HTP. The {\HiLasso} algorithm shows a similar behaviour as the {\HTP} algorithm but requires even more samples to achieve comparable recovery rates. 

\begin{figure}[tb]
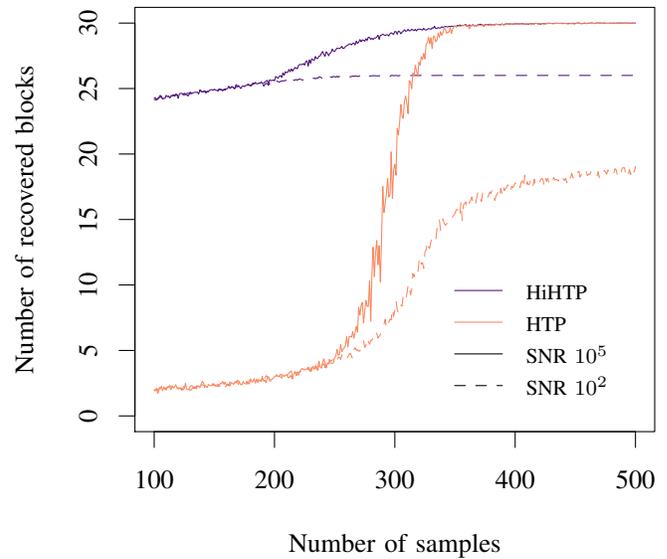

	\centering
	\input{plotGaussianNoisy}
	\hspace{-.2cm}
	\input{plotGaussianNoisyError}
	\vspace{-.6cm}
	\caption{Left: Number of recovered blocks over the number of measurements $m$ for {\htp} and {\HiHTP} in the presence of additive Gaussian noise. The solid line and dashed lines indicate the number of recovered blocks for an $\text{SNR}$ of $10^{5}$ and $10^{2}$, respectively. The signals consist of $N=30$ blocks with $s=4$ blocks having non-vanishing real entries. Right: The mean distance of the reconstructed blocks from the original blocks of the signal  in Euclidean norm over the number of measurements $m$ for {\htp} and {\HiHTP} in the presence of additive Gaussian noise. The solid line and dashed lines indicate the number of recovered blocks for an $\text{SNR}$ of $10^{5}$ and $10^{2}$, respectively. The signal model has dimensions $N=30$, $n=100$, $s=4$ and $\sigma=20$.
	\label{fig:samplingComplexityNoisyGauss}  }
	\vspace{-.5cm}
\end{figure}

The better performance of {\HiHTP} compared to {\HTP} can also be observed if Gaussian noise is added to the measurement vector.  
The block recovery rates for different numbers of samples and signal-to-noise ratios ($\text{SNR}$) of $10^{5}$ and $10^{2}$ are shown in Fig.~\ref{fig:samplingComplexityNoisyGauss} on the left. 
We require a successfully recovered block to deviate by less than $10^{-2}$ in Euclidean norm from the corresponding block of the original signal. 
For $\text{SNR} = 10^{2}$ the noise prohibits the accurate recovery of the active blocks for both algorithms. 
But for {\HiHTP} the correct identification of the zero blocks can still be achieved. 
On the right, Fig.~\ref{fig:samplingComplexityNoisyGauss} shows the mean distance in Euclidean norm of the reconstructed blocks from the original blocks of the signal.

Table~\ref{tab:runtimes} displays the run times per reconstruction of the {\HTP} and {\HiHTP} algorithm on desktop hardware (2 x 2,4 GHz Quad-Core Intel Xeon, 64 GB 1066 MHz DDR3 ECC) in the tests with noiseless Gaussian measurements. The run times are average values of $100$ trials. Both {\HTP} and {\HiHTP} show comparable run times. 

\begin{table}[tb]
	\centering
	\caption{Run times per reconstruction in seconds of HTP and HiHTP on desktop hardware for different numbers of measurements $m$. \label{tab:runtimes}}
	\begin{tabular}{l r r r} \toprule
		$m$ & $200$ & $300$ & $400$ \\ \midrule
		{\HTP}  & $0.29$s & $0.35$s & $0.42$s \\
		{\HiHTP}  & $0.34$s & $0.37$s & $0.40$s \\ \bottomrule
	\end{tabular}
	
\end{table}

\paragraph{Fourier measurements}

In the application to CCRA one aims at recovering $(s,\sigma)$-sparse complex vectors from partial Fourier measurements~\eqref{eq:CCRAMeasurement}.  Fig.~\ref{fig:samplingComplexityFourier} shows the number of successfully recovered blocks from uniform randomly selected Fourier coefficients for {\HiHTP} and {\HTP}.

 An $(s=3,\sigma=10)$-sparse support is drawn uniformly of $N=20$ blocks of dimension $n=50$. The signal entries are complex numbers with real and imaginary part i.i.d.\ sampled from a standard normal distribution. Recovery rates are again averaged over $100$ runs.  Running {\HiHTP} on random partial Fourier measurements shows a qualitatively similar behaviour as for Gaussian measurements.  

 When we select only the $m$ lowest Fourier coefficients instead of a uniformly sampled subset, we observe that the block support is still recovered from a small amount of samples. In contrast, a correct reconstruction of the content of the active blocks requires a comparatively large amount of samples. Both algorithm {\HiHTP} and {\HTP} perform approximately the same using the lowest Fourier modes (not shown in the figures). 

\begin{figure}[tb]
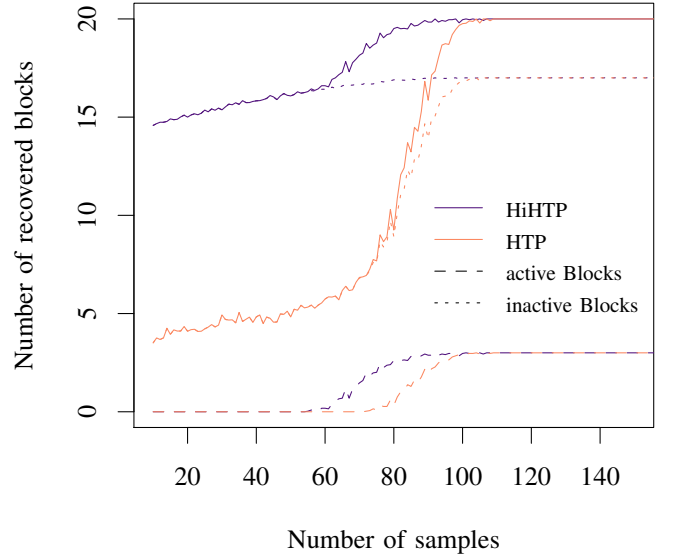

    \centering
    \include{plotFourier}
    \vspace{-.6cm}
    \caption{Number of recovered blocks over the number of measurements $m$ for {\htp} and {\HiHTP} employing uniformly random partial Fourier measurements.  The dashed and dotted lines indicate the average number of correctly recovered zero and non-zero blocks, respectively. The solid lines show the total average  number of recovered blocks. 
    The signals consist of $N=20$ blocks with $s=3$ blocks having non-vanishing real entries. \label{fig:samplingComplexityFourier}}
    \vspace{-.5cm}
\end{figure}

This observation is in agreement with the intuition that the lowest Fourier modes encode information over the large scale structure of the signal.  The information whether a block is active or not can be regarded as a property which does not require to resolve the scale of every entry but only larger blocks of the signal. On the contrary, information on the entries inside a block are encoded in the higher Fourier modes.


\section{General hierarchically sparse framework}\label{sec:generalHierachicalSparsity}
In this section, we provide the general definition of hierarchical sparse vectors that can be efficiently constructed following the strategy of \HiHTP. The sparsity structure can have the general hierarchy of an arbitrary rooted tree, possibly with different block sizes and corresponding sparsities.  We sketch the corresponding algorithm and give a general version 
 of the recovery results of Section~\ref{sec:analytic}.  

\subsection{Setting and notation}

 We consider a rooted tree $T = (V,E)$ with vertex set $V$ and edges $E$. Fig.~\ref{fig:HiSparse} illustrates the following definitions.  We denote the root element by $v_0 \in V$. 
Since the root implies an ordering of the vertices away from the root we get the common notions of the \emph{parent}, \emph{children} and \emph{siblings} of the tree. 
We denote the parent node of $v$ by $\parent(v)\in V$ and set of its children by $\children(v) \subset V$. 
Furthermore, the subset of $V$ that has no children are called $\leaves(T)$. 
We denote the number of leaves of $T$ by $d\coloneqq |\leaves(T)|$.

We can choose an identification $i: \leaves(T) \to [d]$ of the leaves of $T$ with the entries of a vector $\x \in \KK^d$.
Thereby, $\x$ gets partitioned into  hierarchically nested blocks. 

One possible way of identifying vector entries with the leaves of $T$ is given by introducing an ordering among the siblings at every vertex.
We can canonically extend this ordering to all vertices of the same depth by requiring that two vertices with different parent vertices inherit the ordering of their parents. 
Enumerating the leaves of $T$ according to the ordering yields an identification with $[d]$. 
In this way, the groups of siblings in $T$ define a hierarchy of nested blocks of the vector $\x \in \KK^d$. 
The entries of the vector are grouped into blocks as the leaves are grouped into siblings of a parent by $T$. These blocks are again grouped into larger blocks specified by the 
ancestry
of the parents of the leaves and so on. 
See Fig.~\ref{fig:HiSparse} for an illustration. Hereinafter, we always assume that such an identification was fixed.
Nonetheless, other identifications of the leaves with vector entries are also admissible. However, they generally do not give rise to an apparent block structure of the vector $\x$ but rather a permutation thereof. 

Via the identification, the support $\Omega$ of a vector that admits a $T$-structure can be viewed as a subset of the leaves, $\Omega \subset \leaves(T)$. 
We will call $v \in \Omega$ an \emph{active vertex} and further call every ancestor node of an active vertex active as well.
Given a support $\Omega \subset \leaves(T)$ it thereby induces a map $\tree{\Omega}: V \to \mathcal{P}(V)$ such that 
\begin{equation}
	\tree{\Omega}(v) \coloneqq \{ w \in \children(v) \mid \text{$w$ is active}\}.
\end{equation}
We refer to $\tree\Omega$ as the hierarchical support of $\x$, as the active vertices of $V$ identified by $\tree\Omega$ correspond to the hierarchically nested blocks containing non-vanishing entries.

Now we would like to allow only a certain number of vertices to be active among siblings for each vertex, corresponding to a restriction of the number of blocks with non-vanishing entries. 
To this end, we define the map $\tree{n}: V \to \NN$ to count the number of children of a given vertex, {i.e.} $\tree{n}(v) \coloneqq |\children(v)|$.
Furthermore, we define a \emph{sparsity} on $T$ as a map $\tree s: V \to \NN$ with $\tree s(v) \leq \tree n(v)$ for all $v\in V$. 
A support $\Omega \subset \leaves(T)$ or the corresponding hierarchical support  $\tree{\Omega}$ on the entire tree is called $\tree s$-sparse if $|\tree{\Omega(v)}| \leq \tree s(v)$ for all $v \in V$. 

In the following it will be convenient to use recursive definitions. To this end, we introduce for $v \in V$ the notation 
$T_v = (V_v, E_v)$ to refer to the maximal subtree of $T$ with root $v$.

 \begin{figure}[tb]
 	\centering
 	\input{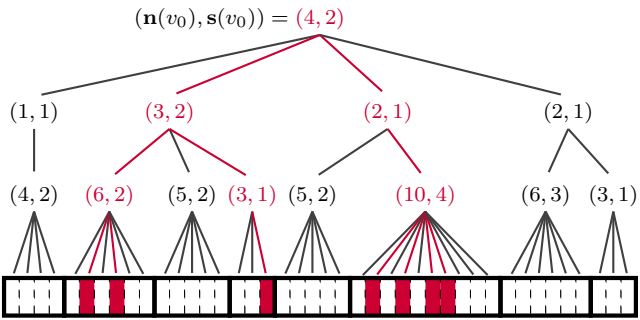}
 	\caption{In this figure, an example of a hierarchical sparse vector is given. The grouping of the entries is encoded in a rooted tree. The children of a vertex constitute a block at their level. The pair of values $(\tree n(v), \tree s(v))$ is indicated at each vertex. The leaves of the tree are identified with the entries of the vector. The support of the vector and  active vertices are highlighted in red.}
 	\label{fig:HiSparse}
 \end{figure}

\subsection{Algorithm and recovery guarantees}

To recover an $\tree s$-sparse vector $\x$ with block structure given by a tree $T$ from linear measurements we can employ the strategy of {\HiHTP}, Algorithm~\ref{alg:HiHTP}, and generalise the $(s,\sigma)$-sparse thresholding operator $L_{s,\sigma}$ to the 
thresholding operator $L^T_{\tree s}(\vec z)$ for
$\tree s$-sparse vectors. 

By the same argument as for two sparsity levels, {i.e.} by the principle of optimal substructures, $L^T_{\tree{s}}(\vec z)$ is the union of the $\tree{s}(v_0)$ many supports $\Omega_w$ of the optimal $\tree{s}|_{V_w}$-sparse approximation to $\vec z$, where $w$ goes through the $\children(v_0)$, which have the highest value of $\norm{{\vec z}_{\Omega_w}}$. Calculating the optimal $\tree{s}|_{V_w}$-sparse approximation of $z$ is naturally done by invoking $L_{T_w, \tree{s}|_{V_w},i}(\vec z)$ for the maximal subtree $T_v$ of $T$ with root $v$. 
This recursive structure is summarised as Algorithm~\ref{alg:LS}.
In Step~\ref{alg:LS:recursion} the Algorithm~\ref{alg:LS} is invoked on each of the subtrees. The actual thresholding step is performed using the standard $s$-sparse thresholding operator $L_s$ in 
Step~\ref{alg:LS:blockthresholding} on the weight vector $\vec \alpha$ that contains one entry per child of $v'_0$. 
The latter operates with a computational complexity $\mathcal{O}(\tree n(v_0'))$. Since this operation needs to be applied in each vertex of $T$, Algorithm~\ref{alg:LS} operates at a computational complexity of $\mathcal{O}(\abs{V})$. 
This is essentially the same as $\mathcal{O}(d)$, i.e. the complexity of applying the standard thresholding operator, if we do not allowed for many vertices in the tree that only have one child. Since the thresholding operation is trivial on only child vertices, this can be assumed without loss of generality. 
Thus, the computation time per iteration of the more general {\HiHTP}-algorithm typically scales  identically to the original {\HTP} Algorithm.

\begin{algorithm}[tb]      
	\caption{Calculation of $L^{T'}_{\tree{s}'}(\vec z)$} 
	\label{alg:LS}
	\begin{algorithmic} [1]
 		\REQUIRE vector $\z$, hierarchical structure $T'$, sparsity $\tree{s}'$
 		\STATE Let $v'_0$ be the root of $T'$ 
 		\FOR {$w \in \children(v'_0)$:}
 			\IF {$w \in \leaves(T')$:}
 				\STATE $\Omega_w = \set{w}$
 				\STATE $a_w = \left|z_{i(w)}\right|$
 			\ELSE
 				\STATE $\Omega_w = L^{T'_w}_{\tree s|_{V'_w}}(\vec z)$ \label{alg:LS:recursion}
 				\STATE $a_w = \|\vec z|_{\Omega_w}\|$
 			\ENDIF
 		\ENDFOR
 		\STATE $\vec a = (a_w \mid w \in \children(v'_0))$
 		\STATE $\Omega^B = L_{\vec s(v'_0)}(\vec a) \subset \children(v'_0)$ \label{alg:LS:blockthresholding}
 		\STATE $\Omega = \bigcup_{w \in \Omega^B} \Omega_w$
 		\ENSURE $\tree s$-sparse support $\Omega$.
	\end{algorithmic}
\end{algorithm}

In complete analogy to $(s,\sigma)$-sparse {\HiHTP}, we can provide the following recovery guarantee for the general $\tree s$-sparse case.

\begin{theorem}[Recovery guarantee]\label{thm:generalRecGarant}
	Suppose that the following RIP condition holds
	\begin{equation}\label{eq:generalRecGarant:RIP}
		\delta_{3\tree s} < \frac{1}{\sqrt{3}}.
	\end{equation}
	Given $\x \in \CC^{d}$, $\e \in \CC^m$, and an $\tree s$-sparse support $\Omega$  on a tree $T$, i.e., the set of indices of the active leaves of $T$. Then, the sequence $(\x^k)$ defined by {\HiHTP} (Algorithm~\ref{alg:HiHTP}) with thresholding operator $L^T_{\tree s}$ and $\y = \A\x\proj_\Omega + \ev$ satisfies, for any $k\geq 0$, 
	\begin{equation}
		\norm{\x^k - \x\proj_\Omega} \leq \rho^k \norm{\x^0 - \x\proj_\Omega} + \tau \norm{\ev},
	\end{equation}
	where 
	\begin{equation} 
		\rho =  \left({\frac{2 \delta_{3\tree s}^2}{1-\delta_{2\tree s}^2}} \right)^{1/2}< 1
	\end{equation}
	and  	$\tau \leq 5.15 /(1-\rho)$.
\end{theorem}
As before we make use of the natural generalization of RIP to the sparsity structure at hand to define ${\tree s}$-RIP constants $\delta_{\tree s}$ bounding $\norm{\A \x}^2$ for all $\tree s$-sparse $\x \in \RR^d$. In addition, for $q\in\NN$ multiplication of $\tree s$ is defined pointwise as $q \tree s(v) \coloneqq \max\{q \tree s(v), \tree n(v)\}$ for all $v \in V$. 

The proof of Theorem~\ref{thm:generalRecGarant} is given together with the proof of Theorem~\ref{thm:recGarant} in Appendix~\ref{app:proofs}. The only required modification is a more general formulation of the Observation~\ref{observation:supportUnions}.

Given an arbitrary $\tree s$-sparse hierarchical support, the result of Ref.\ \cite{BluDav09} allows to provide a bound on $\delta_{\tree s}$ for measurement matrices with {i.i.d.} real Gaussian entries. To this end, the number $L=L(v_0)$ of $\tree s$-sparse hierarchical supports on $T$ can be recursively calculated as using
\begin{equation}
	L(v) = \sum_{\substack{W \subset \children(v),\\ |W|=s(v)}} \prod_{w\in W} L(w) 
\end{equation}
and $L(w) = 1$ for all $w \in \leaves(T)$.

Since the resulting expression is not concise, we illustrate the generalized RIP bound with another important special case. Consider an $\tree s(v)$-sparse setting where all vertices of the depth $i$ share a common number of children $n_i = \tree n(v)$ and sparsity $s_i = \tree s(v)$. The number of such supports of depth $l$ is given by 
\begin{equation}
	L = \prod_{i = 0}^l \binom{n_i}{s_i}^{\prod_{j=-1}^{i-1}s_j}
\end{equation}
with the convention that $s_{-1} \coloneqq 1$. With \eqref{eq:bludav:samplingcomplexity} we therefore get the following generalized version of Theorem~\ref{thm:gaussianBscRIP}: 
\begin{theorem}[$\tree s$-sparse RIP for real Gaussian measurements]\label{thm:generalGaussianBscRIP}
	Let $\A$ be an $m \times n_0\cdot n_1 \cdots n_l$ real matrix with {i.i.d.} Gaussian entries and $m < \prod_{i=0}^l n_i$
	.
	For $\epsilon, \delta > 0$, assume that 
	\begin{equation}
		m \geq \frac{2}{c \delta^2} \left( \sum_{i =0}^l   \prod_{j=-1}^{i-1}s_j \ln \left( \frac{\e\, n_i}{s_i}\right) + \prod_i s_i\ln\left(\frac{12}{\delta}\right) + \ln(\epsilon^{-1}) \right),
	\end{equation}
	with $s_{-1} \coloneqq 1$ and $c>0$ a universal constant.
	Then, with probability of at least $1 - \epsilon$, the restricted isometry constant $\delta_{\tree s}$ of $\A/\sqrt{m}$ satisfies
	\begin{align}
		\delta_{\tree s} \leq \delta. 
	\end{align}
\end{theorem}

More general $\tree s(v)$-sparse settings can be evaluated following the same strategy. However, the resulting expression for the number of samples $m$ are in general more involved. 
Also the result of Theorem~\ref{thm:hirip} for Kronecker product measurements can be readily generalized to hierarchical sparsity with more than two levels: 
\begin{theorem} \label{thm:HiRIPMoreLevels}
  Let $\vec A \in \mathbb{K}^{m\times n_1}$ be a matrix with  RIP-constant $\delta_{s_1}(\vec{A})$ and $\vec B \in \KK^{M\times N}$ one with HiRIP constant $\delta_{s_2, \dots, s_L}(\vec B)$. Then the hierarchical RIP-constant $\delta_{s_1, \dots, s_n}$ of $\vec A \otimes \vec B \in \KK^{mM\times nN}$ satisfies
  \begin{align*}
    \delta_{s_1, \dots , s_L}\leq  \delta_{s_1}(\vec{A}) + \delta_{s_2, \dots, s_L}(\vec{ B}) + \delta_{s_1}(\vec{A})  \cdot \delta_{s_2, \dots, s_L}(\vec{ B}).
  \end{align*}
  In particular, through induction, we obtain for matrices $\vec A_i \in \KK^{m_i\times n_i}$, $i=1, \dots, L$, with $s_i$-th RIP constants $\delta_{s_i}^{\vec A_i}$:
  \begin{align*}
    \delta_{s_1, \dots, s_L}(\vec A_1 \otimes \dots \otimes \vec A_L) \leq \prod_{i=1}^L (1+ \delta_{s_i}({\vec A_i})) -1.
  \end{align*}
\end{theorem}

The techniques of the proof of Theorem \ref{thm:hirip} can readily be adapted to prove also Theorem \ref{thm:HiRIPMoreLevels}. 

\begin{proof}[Proof of Theorem~\ref{thm:HiRIPMoreLevels}]
The proofs reads exactly as the proof of Theorem~\ref{thm:hirip} up to equation \eqref{eq:crucialStep1}, with an adapted version of the flipping operator. 
To be precise, the latter is defined by linearly extending the mapping
\begin{align*}
	\vec{e}_{i_1} \mapsto (\vec{e}_{i_2} \otimes \dots \otimes \vec{e}_{i_L}) =   (\vec{e}_{i_2} \otimes \dots \otimes \vec{e}_{i_L}) \otimes \vec{e}_{i_1}
\end{align*} 
rather than $\vec{e}_i \otimes \vec{e}_j \to \vec{e}_j \otimes \vec{e}_i$. It maps the $(s_1, \dots, s_L)$-sparse vector $\vec{x}$ into a vector of $s_1$-sparse blocks $\vec{h}_i$, for which the estimate \eqref{eq:crucialStep1} can be applied. Flipping back, we obtain a vector of blocks $\x_i$  that are not only $\sigma$-sparse, but rather $(s_2, \dots, s_L)$-sparse. We may hence apply \eqref{eq:crucialStep} with the corresponding RIP of $\B$. 
The statement now follows. 
\end{proof}

\section{Conclusion}

	In this paper, efficient recovery of $(s,\sigma)$-sparse vectors using a variant of {\HTP}, namely the {\HiHTP} algorithm, was formulated. The {\HiHTP} algorithm analytically and numerically proves itself more succesful in the recovery in terms of the sampling complexity compared to established compressed sensing methods. At the same time, 
	it is computationally not more expensive than the original {\HTP} algorithm. 
	The same strategy can be applied to the larger class of hierarchical sparse signals with multiple nested levels of groupings and different sparsities associated to each group. We have also demonstrated the relevance of considering such sparsity structures in several recent applications.

	It is straight-forward to include a collaborative, jointly-sparse extension of the hierarchical sparse structure in the algorithm, where different copies or blocks of a signal have the same sparsity pattern.  To this end, the thresholding operator has to be evaluated on the $\ell_2$-norm of the entries that share the same sparsity structure. In this way, the slightly more general setting originally investigated in Refs.\ \cite{SprechmannEtAl:2010, SprechmannEtAl:2011} can be incorporated in the {\HiHTP} algorithm.

	In this work, bounds on generalized RIP constants for hierarchical sparse structures were derived for Gaussian measurements. It is an interesting open question to find similar results for other measurement ensembles such as partial Fourier measurements. Especially, for an information theoretic analysis of the CCRA framework using {\HiHTP} for the channel estimation task such results are of essential importance. 

	Furthermore, we have shown that a Kronecker product $\vec A \otimes \vec B$ has the HiRIP property
	as soon as its components exhibit the RIP. The analogous result holds for Kronecker products with multiple factors and multi-level hierarchically sparse vectors. This is in contrast to the standard $s$-RIP, where each component needs to have the $s$-RIP. 
	The framework can be directly applied to derive recovery guarantees for a variety of applications, e.g. in massive machine-type communications,  massive MIMO settings, and more generally sparse matrix sketching. 

	It is also an interesting further direction to identify recovery guarantees in the non-commutative analogous setting. 
	The closest such setting is given by \emph{tree tensor networks} \cite{ShiDuanVidal:2006} from quantum mechanics, which are also referred to by the name \emph{hierarchical Tucker} tensor format \cite{HackbuschKuehn:2009, Grasedyck:2010}. 
	Recently, a tensor version of the Iterative Hard Thresholding algorithm \cite{BlumensathDavies:2008} has been put forward covering this non-commutative analogue and already including partial recovery guarantees \cite{RauSchSto16}.

\appendices 

\section{Proof of Theorem~\ref{thm:recGarant} and Theorem~\ref{thm:generalRecGarant}}
\label{app:proofs}

For the proof of Theorem~\ref{thm:recGarant} and Theorem~\ref{thm:generalRecGarant} we need a slightly more general formulation of some standard results about RIP constants to cover the case of the here introduced RIP constants for hierarchically sparse vectors. 
We begin with the following simple statement, which is related to the common equivalent formulations of RIP constants \cite[Chapter~6.1]{FouRau13}. 

\begin{proposition}[Equivalent characterizations of RIP] \label{prop:characterisationsRIP}
	Let $\Omega \subset [d]$ be a support and $\A\in \KK^{m\times d}$ be a measurement map. 
	Then the following two statements are equivalent:
	\begin{enumerate}
		\item $\delta \geq \norm{\Id - \A^\ast_\Omega \A_\Omega}$,
		\item $\forall \x \in \KK^{d}$ with $\supp(\x) \subset \Omega$ 
		\begin{equation}\label{eq:Omega-RIP}
			(1-\delta) \norm{\x}^2 \leq \norm{\A \x}^2 \leq (1+\delta) \norm{\x}^2 .
		\end{equation}
	\end{enumerate}
\end{proposition}

\begin{proof}
	The inequality
	\begin{align}
	\delta \geq \norm{\Id - \A^\ast_\Omega \A_\Omega} 
	&=
	\max_{\x \in \KK^{d}} \frac{|\braket \x {\x - \A^\ast_\Omega \A_\Omega(\x)}|}{\braket \x \x}
	\\
	&=
	\max_{\supp(\x)=\Omega} \frac{|\braket\x\x - \braket{\A\x}{\A\x}|}{\braket{\x}{\x}}
	\end{align}
	holds if and only if for all $\x \in \KK^{d}$ with $\supp(\x) = \Omega$
	\begin{equation}
		\delta \braket \x\x \geq |\braket\x\x - \braket{\A\x}{\A\x}|.
	\end{equation}
	The last bound is equivalent with \eqref{eq:Omega-RIP}.
\end{proof}

If a matrix satisfies RIP, then one can put a similar bound on its adjoint, which can be formalized as an obvious generalization of Ref.\ {\cite[Lemma~6.20]{FouRau13}}:

\begin{proposition}[Adjoint RIP]
	\label{prop:adjointRIP}
	Let $\Omega \subset [d]$ be a support, $\A\in \KK^{m\times d}$ be a measurement map and $\ev \in \KK^d$ a vector.
	If $\snorm{\Id - \A_\Omega^\ast \A_\Omega} \leq \delta$ then 
	\begin{equation}
		\|(\A^\ast \ev)_\Omega\| \leq \sqrt{1+\delta} \, \norm{\ev} \, .
	\end{equation}
\end{proposition}

\begin{proof}
	We use that $\norm{\x_\Omega} = \norm{\x\proj_\Omega}$ 
	and the Cauchy-Schwarz inequality to obtain 
	\begin{align}
	\norm{\right(\A^\ast \ev\left)_\Omega}^2 
	&= \braket{\A^\ast \ev}{\left(\A^\ast\ev\right)\proj_\Omega} 
	= \braket{\ev}{\A\left(\A^\ast\ev\right)\proj_\Omega} 
	\\
	&\leq  \norm{\A(\A^\ast\ev)\proj_\Omega} \norm{\ev} \, .\nonumber
	\end{align}
	Applying Proposition~\eqref{prop:characterisationsRIP} yields
	\begin{equation}
		\norm{\right(\A^\ast \ev\left)_\Omega}^2 
		\leq 
		\sqrt{1+\delta} \norm{\right(\A^\ast\ev\left)_\Omega}\norm{\ev}
	\end{equation}
	and cancellation of the factor $\norm{\right(\A^\ast\ev\left)_\Omega}$ completes the proof. 
\end{proof}

Next, we make an observation allowing to restrict the columns of a matrix. 

\begin{proposition}[Restricting columns]\label{prop:twoSupportsNorm}
	Let $\A\in \KK^{m\times d}$ be a matrix, $\x \in \KK^d$ a vector and $\Omega \subset [d]$ an index set. 
	Then
	\begin{equation}
		\normb{\bigl((\Id-\A^\ast\A)\x\bigr)_\Omega} 
		\leq 
		\snormb{\Id - (\A_T)^\ast\A_T} \norm{\x_T},
	\end{equation}
	where $T = (\supp{\x}) \cup \Omega$.
\end{proposition}

The proof is analogous to the one presented in Ref.\ \cite[Lemma~6.16]{FouRau13}.

\begin{proof}
	With $\z \coloneqq (\Id-\A^\ast\A)\x$, $X\coloneqq \supp(\x)$, and the definition of the projection onto support sets \eqref{eq:def_proj} we obtain
	\begin{align}
	\norm{\bigl(\left(\Id - \A^\ast\A\right)\x\bigr)_\Omega}^2 
	&= \bigl\langle \z\proj_\Omega,  (\Id-\A^\ast\A)\x\bigr\rangle
	\\
	&=  \bigl\langle \z\proj_\Omega, \x \bigr\rangle
	-
	\bigl\langle \A (\z\proj_\Omega) , \A\x \bigr\rangle 
	\\ \nonumber
	&=  \bigl\langle (\z\proj_\Omega)_T, \x_T \bigr\rangle
	-
	\bigl\langle \A_T \z_T , \A_T\x_T \bigr\rangle 
	\\ \nonumber
	&=
	\bigl\langle 
	(\z\proj_\Omega)_T, \bigl(\Id - (\A_T)^\ast\A_T\bigr)\x_T
	\bigr\rangle
	\\
	&\leq
	\norm{\z_\Omega} \snorm{\Id - (\A_T)^\ast\A_T} \norm{\x_T}.
	\end{align}
	Cancelling the factor 
	$\norm{\bigl(\left(\Id - \A^\ast\A\right)\x\bigr)_\Omega}=\norm{\z_\Omega}$
	completes the proof. 
\end{proof}

The proof of Theorem~\ref{thm:recGarant} and Theorem~\ref{thm:generalRecGarant} require a bound for terms involving sums of two and more $(s,\sigma)$-sparse, or general hierarchically $\tree s$-sparse vectors, respectively. To this end we will use the Observation~\ref{observation:supportUnions} from Section~\ref{sec:analytic} for the case of $(s,\sigma)$-sparse vectors and its generalization to general hierarchically $\tree s$-sparse vectors introduced in Section~\ref{sec:generalHierachicalSparsity}. Using the notation introduced in Section~\ref{sec:generalHierachicalSparsity} and following the logic of Section~\ref{sec:analytic} yields the more general observation:
\begin{observation}[Support unions]\label{observation:supportUnions:general}
	For $i=1,2$ let $\Omega_i \subset [d]$ be an $\tree s$-sparse support and $\A\in \KK^{m\times d}$
	with RIP constants $\delta_{\tree s}$. Then
	\begin{equation}\label{eq:supportUnionsAgain}
		\norm{\Id- \A^\ast_{\Omega_1 \cup \Omega_2}\A_{\Omega_1 \cup  \Omega_2}}
		\leq 
		\delta_{2\tree s}.
	\end{equation}
\end{observation}

We are now in the position to prove that {\HiHTP} converges to the ``right'' solution for sufficiently many measurements.

\begin{proof}[Proof of the recovery guarantees, Theorem~\ref{thm:recGarant} and Theorem~\ref{thm:generalRecGarant}]
	We modify the argument as presented for {\htp} in Ref.\ \cite[Proof of Theorem~6.18]{FouRau13}. 
	Similar versions can be found in Refs.\ \cite{Foucart:2011} and \cite{BouchotEtAl:2016}. 
	Most steps of the proofs of Theorem~\ref{thm:recGarant} and Theorem~\ref{thm:generalRecGarant} are identical and are, thus, given here together. 
	We will later explicitely give both statements of the required RIP-type bounds for $(s,\sigma)$-sparse and $\tree s$-sparse vectors, respectively.
	The proofs rely on two observations, that follow directly from the definition of the algorithm: 
	\paragraph{Consequence of Algorithm~\ref{alg:HiHTP}~Line~\ref{alg:HiHTP:TH}}
	By definition, the thresholding operator $L_{\tree{s}}(\z)$  yields the $\vec{s}$-sparse support such that $\norm{\z_{L_{\tree{s}}}(\z)} \geq \norm{\z_\Sigma }$ for all $\tree{s}$-sparse $\Sigma$. 
	This of course also holds for the special case of $L_{s,\sigma}$.
	In particular, Line~\ref{alg:HiHTP:TH} of Algorithm~\ref{alg:HiHTP} (and likewise its generalization) ensures that in the $(k+1)$\textsuperscript{th} iteration
	\begin{equation}
		\norm{\left(\x^k + \A^\ast (\y - \A \x^k)\right)_{\Omega^{k+1}}} \geq \norm{\left(\x^k + \A^\ast (\y - \A \x^k)\right)_\Omega} .
	\end{equation}
	This bound still holds after eliminating the contribution of entries with indices in $\Omega\cap \Omega^{k+1}$. 
	Hence, 
	\begin{equation}
		\begin{split}
			\bigl\|(\x^k &+ \A^\ast (\y - \A \x^k))_{\Omega^{k+1}\setminus \Omega}\bigr\| \\
			&\geq \norm{\left(\x^k + \A^\ast (\y - \A \x^k)\right)_{\Omega \setminus \Omega^{k+1}}} \\ 
			&= \norm{\left(\x_\Omega - \x^{k+1}  + \x^k -\x_\Omega + \A^\ast (\y - \A \x^k)\right)_{\Omega \setminus \Omega^{k+1}}} \\
			&\geq \norm{\left(\x_\Omega - \x^{k+1}\right)_{\Omega\setminus \Omega^{k+1}}} \\
			&\quad\ - \norm{\left( \x^k -\x_\Omega + \A^\ast (\y - \A \x^k)\right)_{\Omega \setminus \Omega^{k+1}}}.
		\end{split}
	\end{equation}
	This inequality can be recast as
	\begin{equation}\label{eq:recGarant:THobservation}
		\begin{split}
			&\phantom{={} } \norm{\left(\x_\Omega - \x^{k+1}\right)_{\Omega\setminus \Omega^{k+1}}}
			\\
			&\leq \norm{\left(\x^k - \x_\Omega + \A^\ast(\y - \A\x^k) \right)_{\Omega^{k+1}\setminus \Omega}} \\
			&\qquad + \norm{\left(\x^k - \x_\Omega + \A^\ast(\y-\A\x^k)\right)_{\Omega\setminus \Omega^{k+1}}} \\
			&\leq \sqrt{2}\, \normB{\left(\x^k - \x_\Omega + \A^\ast(\y - \A\x^k)\right)_{\Omega\Delta\Omega^{k+1}}},
		\end{split}
	\end{equation}
	where $\Omega\Delta\Omega^{k+1} \coloneqq (\Omega \setminus \Omega^{k+1}) \cup (\Omega^{k+1} \setminus \Omega)$ denotes the symmetric difference of $\Omega$ and $\Omega^{k+1}$. 
	
	\paragraph{Consequence of Algorithm~\ref{alg:HiHTP}~Line~\ref{alg:HiHTP:LS}}
	The second observation is that, by definition, $\x^{k+1}$ calculated in 
	Line~\ref{alg:HiHTP:LS} of Algorithm~\ref{alg:HiHTP} fulfills the minimality condition
	\begin{equation}\label{eq:recGarant:LSobservation}
		\left(\A^\ast(\y-\A\x^{k+1})\right)_{\Omega^{k+1}} = 0 ,
	\end{equation}
	which can be seen by setting the gradient of 
	$\x \mapsto \norm{\A \x - \y}^2$ to zero. 
	
	With the bound \eqref{eq:recGarant:THobservation} and the minimality condition \eqref{eq:recGarant:LSobservation} we are now in a 
	position to find the bound
	\begin{equation}
		\begin{split}
			\bigl\Vert\x_\Omega &- \x^{k+1}\bigr\Vert^2
			\\ 
			&= \norm{\left(\x^{k+1} - \x_\Omega\right)_{\Omega^{k+1}}}^2
			\\ 
			&\qquad + \norm{(\x^{k+1}-\x_\Omega)_{\Omega\setminus \Omega^{k+1}}}^2 \\
			&\leq \norm{\left(\x^{k+1} - \x_\Omega + \A^\ast (\y - \A \x^{k+1})\right)_{\Omega^{k+1}}}^2 \\
			&\qquad + 2 \norm{\left(\x^{k} - \x_\Omega +  \A^\ast (\y - \A \x^{k})\right)_{\Omega\Delta\Omega^{k+1}}}^2.
		\end{split}
	\end{equation}
	Inserting $\y = \A \x_\Omega + \ev$ yields
	\begin{equation}\label{eq:recGarant:boundWithFourTerms}
		\begin{split} 
			\bigl\Vert\x^{k+1}&-\x_\Omega\bigr\Vert^2 
			\\
			&\leq 
			\Big[\normb{\left((\Id-\A^\ast \A)(\x^{k+1}-\x_\Omega)\right)_{\Omega^{k+1}}} \\
			&\qquad + \normb{(\A^\ast\ev)_{\Omega^{k+1}}}\Big]^2 \\
			&\qquad+ 2 \Big[\normb{\left((\Id - \A^\ast\A)(\x^k - \x_\Omega)\right)_{\Omega\Delta\Omega^{k+1}}} \\
			&\qquad\qquad+\normb{(\A^\ast\ev)_{\Omega\Delta\Omega^{k+1}}}\Big]^2 .
		\end{split}
	\end{equation}
	\newcommand{\abound}{\ensuremath c_1}
	\newcommand{\bbound}{\ensuremath c_3} 
	\newcommand{\cbound}{\ensuremath c_2}
	\newcommand{\dbound}{\ensuremath c_4}
	Using Proposition~\ref{prop:twoSupportsNorm} and Observation~\ref{observation:supportUnions} 
	we find for the first and third term that
	\begin{equation}\label{eq:recGarant:firstTermBound}
		\begin{split} 
			\bigl\|\big((\Id-&\A^\ast \A)(\x^{k+1}-\x_\Omega)\big)_{\Omega^{k+1}}\bigr\| \\
			&\leq \snormb{\Id-\A_{\Omega \cup \Omega^{k+1}}^\ast\A_{\Omega \cup \Omega^{k+1}}}\norm{\x^{k+1}-\x_\Omega} \\
			&\leq \abound \normb{\x^{k+1}-\x_\Omega} 
		\end{split}
	\end{equation}
	and
	\begin{equation}\label{eq:recGarant:thirdTermBound}
		\begin{split} 
			\bigl\|\big((&\Id - \A^\ast\A)(\x^k - \x_\Omega)\big)_{\Omega\Delta\Omega^{k+1}}\bigr\| \\ 
			&\leq \snormb{\Id-\A_{\Omega \cup (\Omega^{k+1}\setminus \Omega) \cup \Omega^k}^\ast
				\A_{\Omega \cup (\Omega^{k+1}\setminus \Omega) \cup \Omega^k}} 
			\normb{\x^k - \x_\Omega} \\
			&\leq \cbound \normb{\x^k-\x^\Omega},  
		\end{split}
	\end{equation}
	where $\abound$ and $\cbound$ are suitable RIP constants. 
	For the case of $(s,\sigma)$-sparse vectors Observation~\ref{observation:supportUnions} ensures that we can choose $\abound \coloneqq \delta_{2s,2\sigma}$. 
	Furthermore, 
	we notice that $\Omega \cup (\Omega^{k+1}\setminus \Omega)$ is $(2s,\sigma)$-sparse and correspondingly define $\cbound \coloneqq \delta_{3s,2\sigma}$. 
	In order to prove the more general version for hierarchical $\tree s$-sparse vectors with more than two hierarchy levels, the analogous argument and Observation~\ref{observation:supportUnions:general} allows us to choose $\abound \coloneqq \delta_{2\tree s}$ and $\cbound \coloneqq \delta_{3 \tree s}$.
	For the remaining two terms Proposition~\ref{prop:adjointRIP} and 
	Observation~\ref{observation:supportUnions} yield
	\begin{equation}
		\norm{(\A^\ast\ev)_{\Omega^{k+1}}} \leq \bbound \norm{\ev} \label{eq:recGarant:secondTermBound}
	\end{equation}
	with $\bbound \coloneqq  ({1+\delta_{s,\sigma}})^{1/2} $ 
	or $\bbound \coloneqq ({1+\delta_{\tree s}})^{1/2}$, respectively, 
	and 
	\begin{equation}
		\norm{(\A^\ast\ev)_{\Omega\Delta\Omega^{k+1}}} \leq \sqrt{1 + \abound} \norm{\ev}. \label{eq:recGarant:fourthTermBound}
	\end{equation}

	Plugging \eqref{eq:recGarant:firstTermBound}, \eqref{eq:recGarant:secondTermBound}, \eqref{eq:recGarant:thirdTermBound}, \eqref{eq:recGarant:fourthTermBound} back into \eqref{eq:recGarant:boundWithFourTerms}, solving the quadratic inequality for $\norm{\x^{k+1} + \x_\Omega}$ and then using that $\sqrt{a^2+b^2} \leq |a| + |b|$ for $a,b \in \RR$ leads to 
	\begin{equation}
		\begin{split}
			\normb{\x^{k+1} + \x_\Omega} &\leq  \Bigg[ \frac{2}{1-\abound^2} \left(\cbound \normb{\x^k - \x_\Omega} + \sqrt{1+\abound} \norm{\ev} \right)^2 \\
			&\quad + \left(\frac{\bbound}{1-\abound^2}\right)^2 \norm{\ev}^2\Bigg]^{1/2} + \frac{\abound\bbound}{1- \abound^2} \norm{\ev} 
			\\
			&\leq \sqrt{\frac{2 \cbound^2}{1-\abound^2}}  \norm{\x^k - \x_\Omega} \\ 
			&\quad +\frac{\sqrt{2 (1-\abound)} + \bbound}{1-\abound} \norm{\ev}.
		\end{split}
	\end{equation}
	
	We define $\rho \coloneqq \sqrt{2 \cbound^2/(1-\abound^2)}$ and observe that 
	$\rho \leq \sqrt{2 \cbound^2/(1-\cbound^2)} < 1 $
	if 
	$\delta_{3s,2\sigma} < 1/\sqrt{3}$ or
	$\delta_{3\tree s} < 1/\sqrt{3}$, respectively
	.
	Furthermore, we define the parameter $\tau$ such that $(1-\rho)\tau = \frac{\sqrt{2 (1-\abound)} + \bbound}{1-\abound} $
	holds. 
	In the regime $0 \leq c_1 \leq 1/\sqrt{3}$ we can make use of the linear bound 
	\begin{equation}\label{eq:linbound}
		\frac{\sqrt{2 (1-\abound)} + \bbound}{1-\abound}  \leq \sqrt{2}+\bbound + \lambda\abound
	\end{equation}
	with $\lambda = (4.733 \bbound + 2.637)/2$. 
	Plugging  
		$
	\bbound \leq \sqrt{1+\cbound} < \sqrt{1+1/\sqrt{3}} $
	and $c_1 \leq c_2 < 1/\sqrt{3}$ into the linear bound \eqref{eq:linbound} yields
	$
		(1-\rho)\tau \leq 5.15,
	$
	which completes the proof. 
\end{proof}

\begin{figure}[tb]
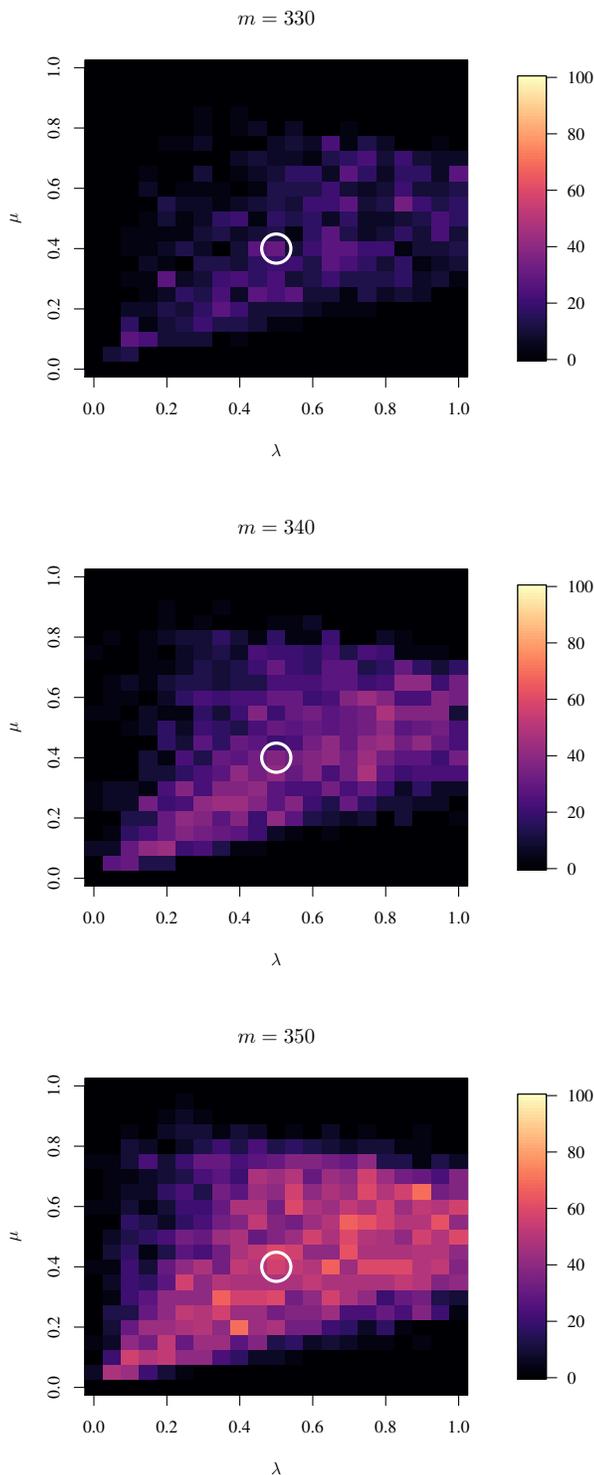

	\centering
	\include{plotCalibration}
	\caption{The figures show the percentage of recovered signals for {\HiLasso} using different calibration parameters $\mu$ and $\lambda$ for numbers of measurements $m=320$, $m=340$ and $m=350$, respectively. The value of $\mu=0.4$ and $\lambda=0.5$ used for Fig.~\ref{fig:samplingComplexityGauss} are highlighted by a white circle. The dimension of the signals are $N=30$, $n=100$, $s=4$ and $\sigma = 20$ }
	\label{fig:HiLassoCalibration}
\end{figure}

\section{Calibration of {\HiLasso} }\label{app:calibration}

In the optimization problem of the {\HiLasso} algorithm 
\begin{equation}
	\operatorname{minimize} \frac{1}{2} \norm{\y-\A\x}^2 + \mu \|\x\|_{\ell_1} + \lambda \| \x \|_{\ell_2/\ell_1}.
\end{equation}
the parameters $\mu,\lambda>0$ have to be chosen according to the sparsity structure of the signal and the noise level of the measurements. 
In our tests the signals are $(s=4, \sigma=20)$-sparse with $N=30$ blocks each of dimension $n=100$. The support is drawn uniformly at random and the entries are i.i.d.\ real numbers drawn from a standard normal distribution. Fig.~\ref{fig:HiLassoCalibration} shows the percentage of recovered signals for different values of the parameters $\mu$ and $\lambda$ of the {\HiLasso} algorithm and different numbers $m$ of noiseless Gaussian measurements. For each combination of $\mu$ and $\lambda$ the algorithm is tested for $30$ signals.  
 A signal is recovered if it deviates by less than $10^{-5}$ in Euclidean norm from the original signal.  
We observe for a number of measurements $m$ between $300$--$350$ that choosing non-zero values for both parameters yields recoveries while setting one parameter to zero, corresponding to standard $\ell_1$-regularization and mixed $\ell_1/\ell_2$-norm regularization, does not.  

For the numerics of Fig.~\ref{fig:samplingComplexityGauss},  we have chosen the parameter $\mu=0.5$ and $\lambda=0.4$. This calibration point lies approximately in the centre of the parameters for which recoveries are observed and is among the maximal points within the statistical error.  In fact, using a couple of randomly selected further calibration points which appear reasonable from Fig.~\ref{fig:HiLassoCalibration} yield the same results for Fig.~\ref{fig:samplingComplexityGauss}.

\section*{Acknowledgements}
We thank A.~Steffens, C.~Riofr\'io, D.~Hangleiter and C.~Krumnow for helpful discussions. We are also grateful to S.~Kiti\'c, P.~Schniter and P.~Boufounos for pointing to related references. 
The research of IR, MK, JE has been supported by 
the DFG project EI 519/9-1 (SPP1798 CoSIP, MATH+),
the Templeton Foundation, the EU (RAQUEL), 
the BMBF (Q.com), and 
the ERC (TAQ)., Q.Link-X).
The work of MK has also been funded by the
	Excellence Initiative of the German Federal and State Governments (Grant ZUK 81), 
	the DFG (SPP1798 CoSIP), and 
 the National Science Centre, Poland within the project Polonez (2015/19/P/ST2/03001) 
 which has received funding from the European Union's Horizon 2020 research and innovation programme under the 
 Marie Sk{\l}odowska-Curie grant agreement No. 665778. 
The work of GW was carried out within DFG grants WU 598/7-1 and WU 598/8-1 (SPP1798 CoSIP), 
and the 5GNOW project, supported by
the European Commission within FP7 under grant 318555.
{AF acknowledges support by the (DFG) Grant KU 1446/18-1 and by ANR
JCJC OMS.}

\bibliographystyle{IEEEtran}

\end{document}